\documentclass[a4paper,twocolumn,10pt,accepted=2023-03-28]{quantumarticle}
\pdfoutput=1

\usepackage[utf8]{inputenc}

\usepackage{fullpage} 
\usepackage{parskip} 
\usepackage{tikz} 
\usepackage{amsmath,amsthm,amsfonts,amssymb}
\usepackage{verbatim}
\usepackage{hyperref}
\usepackage{float}
\usepackage{bm}
\usepackage{graphicx}
\usepackage{geometry}
\usepackage{wrapfig}

\usepackage{marginnote}
\usepackage{multirow}

\reversemarginpar

\theoremstyle{definition}
\newtheorem{corollary}{Corollary}
\newtheorem{definition}{Definition}

\newtheorem{lemma}{Lemma}
\newtheorem{proposition}{Proposition}
\newtheorem{theorem}{Theorem}

\newtheorem*{remark}{Remark}

\newcounter{mycounter}
\newtheorem{lemma-Appendix}[mycounter]{Lemma}

\newcommand{\mbf}{\mathbf}
\newcommand{\mbb}{\mathbb}
\newcommand{\mrm}{\mathrm}
\newcommand{\mc}{\mathcal}

\newcommand{\tr}{\textrm{Tr}}

\newcommand{\ket}[1]{|#1\rangle}
\newcommand{\bra}[1]{\langle #1|}
\newcommand{\op}[2]{|#1\rangle\langle #2|}

\newcommand{\pd}[2]{\frac{\partial #1}{\partial #2}}
\newcommand{\x}{\mathbf{x}}

\newcommand{\rnd}{\text{rnd}}

\newcommand{\deltac}{\overline{\Delta \zeta}}
\newcommand{\deltacs}{\overline{\Delta \zeta^{\star}}}
\newcommand{\T}{\text{T}}
\newcommand{\tot}{\text{tot}}

\definecolor{cool_green}{rgb}{0.0, 0.5, 0.0}

\definecolor{navy_blue}{rgb}{0,0,0.5}

\begin{document}

\title{The Round Complexity of Local Operations and Classical Communication (LOCC) in Random-Party Entanglement Distillation}

\author{Guangkuo Liu}
\affiliation{Department of Physics, University of Illinois at Urbana-Champaign, Urbana, Illinois  61801,USA}
\affiliation{JILA, University of Colorado/NIST, Boulder, CO, 80309, USA}
\affiliation{Department of Physics, University of Colorado, Boulder CO 80309, USA} 
\author{Ian George}
\affiliation{Department of Electrical and Computer Engineering, Coordinated Science Laboratory, University of Illinois at Urbana-Champaign, Urbana, Illinois 61801, USA}
\affiliation{Illinois Quantum Information Science and Technology (IQUIST) Center,
University of Illinois Urbana-Champaign, Urbana, IL 61801}
\author{Eric Chitambar}
\affiliation{Department of Electrical and Computer Engineering, Coordinated Science Laboratory, University of Illinois at Urbana-Champaign, Urbana, Illinois 61801, USA}
\affiliation{Illinois Quantum Information Science and Technology (IQUIST) Center,
University of Illinois Urbana-Champaign, Urbana, IL 61801}

\date{March 28, 2023}

\begin{abstract}
    A powerful operational paradigm for distributed quantum information processing involves manipulating pre-shared entanglement by local operations and classical communication (LOCC).  The LOCC round complexity of a given task describes how many rounds of classical communication are needed to complete the task.  Despite some results separating one-round versus two-round protocols, very little is known about higher round complexities.  In this paper, we revisit the task of one-shot random-party entanglement distillation as a way to highlight some interesting features of LOCC round complexity.  We first show that for random-party distillation in three qubits, the number of communication rounds needed in an optimal protocol depends on the entanglement measure used; for the same fixed state some entanglement measures need only two rounds to maximize whereas others need an unbounded number of rounds.  In doing so, we construct a family of LOCC instruments that require an unbounded number of rounds to implement.  We then prove explicit tight lower bounds on the LOCC round number as a function of distillation success probability.  Our calculations show that the original W-state random distillation protocol by Fortescue and Lo is essentially optimal in terms of round complexity.
\end{abstract}

\maketitle

\section{Introduction}

Quantum entanglement is an essential ingredient for realizing the full capabilities of distributed quantum information processing.  To understand the research character of entanglement, one typically considers the ``distant-lab'' paradigm in which spatially separated laboratories have access to different parts of a globally-entangled states, but they can only process it using local operations and classical communication (LOCC) \cite{Plenio-2007a, Horodecki-2009a}.   In this scenario, parties take turns performing a local quantum measurement, announcing the outcome, and choosing a new local measurement for the next round based on the global history of previous results.  Every LOCC protocol then has a tree-like structure with each branch in the tree representing a different sequence of measurement outcomes.  Unfortunately, for many tasks such as entanglement distillation \cite{Bennett-1996a, Devetak-2005a} or state discrimination \cite{Peres-1991a, Bennett-1999a, Kleinmann-2011a, Childs-2013a, Chitambar-2013a}, precise bounds on LOCC capabilities are hard to prove.  One reason for this difficulty is that a potentially unbounded number of classical communication exchanges are allowed in an LOCC protocol.  

The largely unexplored topic of LOCC round complexity studies how many rounds of classical communication are needed to perform some distributed quantum information task.  We will say that an LOCC protocol has $r$ rounds if there is at least one branch having a sequence of local measurements alternating $r$ times between different parties (see Section \ref{Sect:lower-bound-EPR} for a slightly more general definition).  Note that each local measurement might involve a combination of measurements all performed in the same laboratory.   Most results on LOCC round complexity involve quantum state discrimination problems and focus just on separating one-round versus two-round success rates \cite{Owari-2008a, Nathanson-2013a, Chitambar-2013a, Chitambar-2014a, Croke-2017a, Tian-2016a, Yuan-2020a, Yang-2021a}.  Much more challenging is proving that a given task requires more than two rounds of LOCC to complete, and only a few results are known of this form.  Wakakuwa \textit{et al.} have proven separation results between two and three rounds for the task of entanglement-assisted nonlocal gate implementation \cite{Wakakuwa-2017a, Wakakuwa-2019a}.  For higher rounds, Wang and Duan have constructed families of bipartite quantum states having increasing dimension but requiring more rounds to perfectly discriminate the states as the dimension grows \cite{Xin-2008a}.  A similar finding has been shown for the convertibility of certain quantum (resp. classical) mixed states \cite{Chitambar-2017a}.  
On the other hand, some round compression results are known.  For example, in the task of LOCC state conversion, any bipartite protocol can always be reduced to one-way LOCC provided the initial state is pure, regardless of the system's dimensions \cite{Lo-2001a, Nielsen-1999a, Vidal-1999a}.  In another direction, Cohen has shown that certain LOCC tasks of high round complexity can be approximated arbitrarily well using one-way LOCC \cite{Cohen-2015a}.   Complementing the work on round or ``depth'' compression, it has recently been shown that the ``width'' of an arbitrary LOCC state discrimination protocol can be compressed to a standard form \cite{Leung-2021a}.  One could in general study the interplay between depth and width complexities in different LOCC tasks.  

The purpose of this work is to further advance the study of LOCC round complexity.  To exemplify certain new facts about round complexity, we focus on the specific task of random-party entanglement distillation in three-qubit W-class states \cite{Fortescue-2007a}, which has already proven to be a fruitful problem for demonstrating interesting properties of LOCC \cite{Cui-2011a, Chitambar-2011a, Chitambar-2012a, Chitambar-2012b}.  A W-class state $\ket{\Psi_W}$ is one that can be obtained from the canonical W state $\ket{W}=\frac{1}{\sqrt{3}}(\ket{100}+\ket{010}+\ket{001})$ by stochastic LOCC (SLOCC) \cite{Dur-2000a}.  A random-party EPR distillation protocol transforms a tripartite W-class state $\ket{\Psi_W}$ into a bipartite EPR state $\ket{\Psi^+}=\frac{1}{\sqrt{2}}(\ket{01}+\ket{10})$ with the target pair unspecified; any branch in the protocol is deemed a success branch provided $\ket{\Psi^+}$ is obtained between some pair of parties at the end of the branch.   Fortescue and Lo devised a family LOCC protocols with increasing round number (details given below) that completes this task with success probabilities approaching one.  The limit of such protocols is then some map that distills $\ket{W}$ into random-party EPR pairs with probability one.  However, it has been shown \cite{Chitambar-2012a} that this map or any map achieving unit success probability is not implementable by LOCC.  Consequently, LOCC constitutes a class of quantum operations that is not closed.  This argument was later extended in Ref. \cite{Chitambar-2014b} to show that the set of bipartite LOCC maps is likewise not topologically closed.  Related topological properties have been reported in Ref. \cite{Childs-2013b, Cohen-2015a, Cohen-2017a}, and we also note that an experimental demonstration of random distillation has recently been performed using a three-photon W-state \cite{Li-2020a}. 

Variations to the random-party EPR distillation task can also be considered.  For example, one can demand that the final states obtained in each success branch be just some entangled bipartite state $\ket{\phi}$, and not necessarily a maximally entangled one \cite{Fortescue-2008a, Chitambar-2011a}.  If $E$ is some measure of bipartite entanglement, the goal then is to maximize the average final value $E$ across all success branches.  More precisely, we define
\begin{equation}
\label{Eq:entanglement-random}
    E^{(\rnd)}(\Psi_W):=\sup_{\Psi_W\to\{p_i,\varphi_i\}}\sum_i p_iE(\varphi_i),
\end{equation}
where the supremum is taken over all finite-round LOCC transformations that convert $\ket{\Psi_W}$ into some bipartite state $\ket{\varphi_i}$ with probability $p_i$.  A convenient entanglement measure for bipartite systems is the concurrence \cite{Wootters-1998a}, which for a pure state $\ket{\varphi}^{XY}$ is defined as $C(\varphi)=2\left(\det[\tr_X\op{\varphi}{\varphi}]\right)^{1/2}$.  Thus, $C^{(\rnd)}(\Psi_W)$ is the largest average pure-state concurrence obtainable from $\ket{\Psi_W}$.  We will refer to the task of optimizing the expected bipartite concurrence as random-party concurrence distillation.  Another entanglement measure is given $E_2(\varphi)=2\lambda_{\min}[\tr_X\op{\varphi}{\varphi}]$, which is twice the smaller eigenvalue of the reduced density matrix $\tr_X\op{\varphi}{\varphi}$.  Operationally, it is known that $E_2(\varphi)$ corresponds to the supremum probability of transforming the state $\ket{\varphi}$ to an EPR state by LOCC \cite{Lo-2001a}.  From this it immediately follows that $E_2^{(\rnd)}(\Psi_W)$ is equal to the supremum probability of obtaining an EPR pair between any two parties starting from $\ket{\Psi_W}$, and so the optimization in Eq. \eqref{Eq:entanglement-random} can be restricted to LOCC transformations with only EPR states being the target, i.e. the original task of random-party EPR distillation.

A second variation to the random-party EPR distillation task is to fix one of the parties ($\star$) and only deem the transformation a success if party $\star$ is entangled in the end.  We refer to this as $\star$-random-party distillation, as opposed to \textit{total} random-party distillation described above.  For a fixed party $\star\in\{A,B,C\}$, the optimal average bipartite entanglement for party $\star$ is
\begin{equation}
\label{Eq:entanglement-star-random}
    E^{(\star\text{-}\rnd)}(\Psi_W):=\sup_{\Psi_W\to\{p_i,\varphi_i\}}\sum_i p_iE(\varphi_i),
\end{equation}
where the supremum is taken over all LOCC transformations in which each state $\ket{\varphi_i}$ is bipartite entangled between party $\star$  and some other party in $\{A,B,C\}\setminus\{\star\}$.  Restricted random-party distillation has been partially studied for the concurrence and $E_2$ entanglement measures in Refs. \cite{Chitambar-2014b} and \cite{Chitambar-2012b}, respectively.

In terms of round complexity, one of our main results is that $E_{2}^{(\star\text{-}\rnd)}(\Psi_W)$ is achievable in finite rounds of LOCC for any tripartite W-state, whereas $C^{(\star\text{-}\rnd)}(\Psi_W)$ is not.  Hence, the amount of rounds needed in an optimal entanglement distillation protocol depends on the type of entanglement measure considered.  What makes this result particularly surprising is that concurrence and $E_2$ measures are in one-to-one correspondence: $E_2(\ket{\varphi})=1+\sqrt{1-C(\ket{\varphi})^2}$.
Our second main result is establishing a tight lower bound on the number of LOCC rounds needed to achieve a random-party EPR distillation of $\ket{W}$ with probability $>1-\delta$, for any $\delta>0$.  To our knowledge, this is the first time any type of trade-off has been obtained between round complexity and success probability of a given distributed quantum information processing task.  This result also shows that the original Fortescue-Lo protocol is essentially optimal for random-party distillation of $\ket{W}$, and we prove optimality of more general types of transformations by enlarging the class of LOCC to encompass all operations that completely preserve positivity of the partial transpose (PPT) \cite{Rains-1999a}.  As a corollary of our work, we introduce a new family of quantum instruments that lie on the boundary of LOCC but not inside LOCC itself \cite{Chitambar-2014b}.

\section{The Structure of W-class States}

\label{Sect:W-class-structure}

In this paper, we refer to the W-class as the collection of three-qubit states that can be obtained from the canonical W state $\ket{W}$ by SLOCC.  Up to a local change in basis, every W-class state can be written as
\begin{align}
\ket{\mbf{x}}\!:=\!\sqrt{x_0}\ket{000}+\!\sqrt{x_A}\ket{100}+\!\sqrt{x_B}\ket{010}+\!\sqrt{x_C}\ket{001},\notag
\end{align}
with the $x_i$ being non-negative.  Normalization requires that $x_0=1-(x_A+x_B+x_C)$, and therefore the state $\ket{\mbf{x}}$ only depends on three non-negative parameter $\mbf{x}=(x_A,x_B,x_C)$.  If all three of these numbers are strictly positive, then the state is genuinely three-way entangled (i.e. it is not a product state with respect to some bi-partition of the three parties).  Moreover, in this case it can be shown that the vector $\mbf{x}$ \textit{uniquely} identifies the state \cite{Kintas-2010a}.  More precisely, if we express $\ket{\mbf{x}}$ in some other basis as \[\sqrt{x_0'}\ket{0'0'0'}+\sqrt{x_A'}\ket{1'0'0'}+\sqrt{x_B'}\ket{0'1'0'}+\sqrt{x_C'}\ket{0'0'1'}\]
then necessarily $x_k=x_k'$ for all components.  Equivalently stated, the components $x_k$ are invariant under local unitaries (LU).  This LU invariance, combined with the fact that the W-class is closed under LOCC, makes the W-class very attractive for studying properties of LOCC.  Indeed, if Alice, Bob, and Charlie start out sharing a W-class state $\ket{\mbf{x}}$, then any multi-outcome LOCC protocol on $\ket{\mbf{x}}$ can be described compactly by a probabilistic transformation of vectors,
\begin{equation}
    \mbf{x}\mapsto \mbf{y}_i\quad\text{with probability $p_i$},
\end{equation}
in which the $\ket{\mbf{y}_i}$ are different W-class states obtained along different branches of the protocol.

Every bipartite entangled state belongs to the W-class, and in canonical form, they have one and only one of the coordinates $\{x_A,x_B,x_C\}$ equaling zero.  For example, $\ket{\varphi}=\sqrt{x_0}\ket{000}+\sqrt{x_A}\ket{100}+\sqrt{x_B}\ket{010}$ is a bipartite entangled state shared between Alice and Bob whenever $x_A,x_B>0$.  Its concurrence is readily computed to be (see Appendix \ref{appendix:a})
\begin{equation}
\label{Eq:bipartite-concurrence}
    C(\varphi)=2\sqrt{x_Ax_B}.
\end{equation}
From this we see that the concurrence is homogeneous under multiplication by a non-negative scalar.  That is, $C(\alpha\varphi)=\alpha C(\varphi)$ for any $\alpha\geq 0$.  We also observe that $\ket{\varphi}$ is maximally entangled iff $x_A=x_B=\frac{1}{2}$ and $x_0=0$.

In an LOCC protocol, each party takes turns measuring their local system and announcing the result.  Since each party holds a qubit system, every local measurement can be described by a set of $2\times 2$ Kraus operators $\{M_i\}_i$.  Without loss of generality, we can assume that each $M_i$ has been brought into upper triangular form by a local unitary.  Indeed, if $U_iM_i$ is upper triangular, then we can always append $U_i^\dagger$ to the start of the next round of local measurement to recover the action of the original Kraus operator $M_i$.  Hence we will write the Kraus operators of every local measurement as
\begin{equation}
    M_i=\begin{pmatrix}\sqrt{a_i}&b_i\\0&\sqrt{c_i}\end{pmatrix},
\end{equation}
where $a_i,c_i\geq 0$.  The completion relation $\sum_{i}M_i^\dagger M_i=\mbb{I}$ requires that $\sum_ia_i=1$ and $\sum_i c_i\leq 1$.  When party $k$ performs a local measurement $\{M_i\}_i$ on the W-state $\ket{x}$ and obtains outcome $i$, then the post-measurement state $\ket{\mbf{y}_i}=\frac{1}{\sqrt{p(i)}}M_i\ket{\mbf{x}}$ has undergone a coordinate transformation
\begin{align}
    x_k\mapsto y_{i,k}&=\frac{c_i}{p_i}x_k,\\
    x_j\mapsto y_{i,j}&=\frac{a_i}{p_i}x_j\qquad \text{for $j\not=k\text{ or } 0$},
\end{align}
where $p_i=\bra{\mbf{x}}M_i^\dagger M_i\ket{\mbf{x}}$.  Observe the monotonicity conditions \cite{Kintas-2010a}
\begin{align}
    \sum_{i}p_i y_{i,k}&=\sum_{i}c_ix_k\leq x_k\label{Eq:coordinate-transformation-measuring}\\
     \sum_{i}p_i y_{i,j}&=\sum_{i}a_ix_j = x_j \qquad\text{for $j\not= k\text{ or } 0$}.\label{Eq:coordinate-transformation-not-measuring}
\end{align}
In other words, the component of the measuring party is non-increasing on average where as the components of the non-measuring parties remain unchanged on average.  

An important property of any local measurement is that it can be decomposed into a sequence of weak measurements \cite{OB05}.  Hence, we can envision a given LOCC W-state conversion as a tree of smooth trajectories in the positive quadrant of $\mbb{R}^3$ taking the initial state $\mbf{x}$ to its possible target states $\{\mbf{y}_i\}_i$.  Consequently, if $x_k\mapsto y_{i,k}$ is the component transformation for party $k$ along one branch of some protocol, then there exists an equivalent sequence of protocols that passes through any desired coordinates between $x_k$ and $y_{i,k}$.   

Finally, we introduce three important functions of three-qubit W-class states that play an important role in this work.  For coordinates $(x_A,x_B,x_C)$ of state $\ket{\mbf{x}}$, let $n_1,n_2,n_3\in\{A,B,C\}$ be distinct party labels such that $x_{n_1}\geq x_{n_2}\geq x_{n_3}$ and let $n_i(\mbf{x}):=x_{n_i}$.  Suppose also that $x_{n_2}>0$ so that $\ket{\mbf{x}}$ is not a product state.  Then define
\begin{align}
    \eta(\mbf{x})&:=x_{n_2}+x_{n_3}-\frac{x_{n_2}x_{n_3}}{x_{n_1}}\notag\\
    \kappa(\mbf{x})&:=2(x_{n_2}+x_{n_3})-\frac{x_{n_2}x_{n_3}}{x_{n_1}}\notag\\
    \zeta(\mbf{x})&:=2\sqrt{x_{n_1}x_{n_2}}+\frac{2}{3}x_{n_3}\sqrt{\frac{x_{n_1}}{x_{n_2}}}+\frac{1}{3}\frac{x_{n_2}x_{n_3}}{x_{n_1}}.
\end{align}
With a slight abuse of the notation, we introduce one additional function.  Let $\star\in\{1,2,3\}$ be any fixed party, and take $n_1,n_2\in\{A,B,C\}\setminus\{\star\}$ as distinct party labels such that $x_{n_1}\geq x_{n_2}>0$.  Then define
\begin{align}
       \zeta^\star(\mbf{x}):=2\sqrt{x_{\star}x_{n_1}}+\frac{2}{3}x_{n_2}\sqrt{\frac{x_{\star}}{x_{n_1}}}.\phantom{xxxxxxxxxxx}
\end{align}

It has been shown that $\eta(\mbf{x})$, $\kappa(\mbf{x})$, and $\zeta^\star(\mbf{x})$ are entanglement monotones when manipulating W-class states under LOCC \cite{Chitambar-2012a, Chitambar-2012b, Chitambar-2014b}; i.e. 
\begin{equation}
    \eta(\mbf{x})\geq \sum_ip_i\eta(\mbf{y}_i)
\end{equation}
for any LOCC transformation $\ket{\mbf{x}}\to\{p_i,\ket{\mbf{y}_i}\}$, and likewise for $\kappa(\mbf{x})$ and $\zeta^\star(\mbf{x})$.  Notice that since the W-class is closed under LOCC \cite{Dur-2000a},  $\ket{\mbf{x}}\to\{p_i,\ket{\mbf{y}_i}\}$ describes the most general type of pure-state transformations possible by LOCC and hence we have monotonicity under arbitrary pure state transformations.  Below we show that the same is true for the $\zeta(\mbf{x})$, and we also reveal its operational meaning in terms of random-party concurrence distillation.

\section{Total Random-Party Distillation}
\label{Sect:Total-Random-Party-Distillation}

For the task of total random-party EPR distillation, the original Fortescue-Lo (F-L) protocol gives $E_2^{(\rnd)}(W)=1$ for the W state $\ket{W}$ \cite{Fortescue-2007a}.  The F-L protocol consists of each party performing the same binary-outcome measurement with Kraus operators
\begin{align}
\label{F-L-Kraus}
     &M_0^\epsilon=\sqrt{1-\epsilon}\op{0}{0}+\op{1}{1}, &M_1^\epsilon=\sqrt{\epsilon}\op{0}{0}.
\end{align}
Consider one iteration of the protocol in which each party performs this measurement.  If all three parties obtain outcome $M_0$, then the post-measurement state is $M_0\otimes M_0\otimes M_0\ket{W}\propto \ket{W}$; the W-state is left on changed, and they all repeat the same measurement in the next iteration of the protocol.  On the other hand if one party obtains outcome $M_1$ while the other two obtain $M_0$, then the latter are left in an EPR state.  The final possibility is that at least two parties obtain outcome $M_1$ which breaks all the entanglement in the state.  For any $\delta>0$, the value $\epsilon$ can be chosen sufficiently small and the number of iterations can be chosen sufficiently large such that an EPR state is obtained with probability $>1-\delta$.  In Section \ref{Sect:lower-bound-EPR} we identify a precise trade-off between the value $\delta$ and the number of rounds needed to obtain a success probability $>1-\delta$.  The F-L protocol was later generalized to other W-class states.  It was shown that 
\begin{equation}
\label{Eq:E2-random}
    E_2^{(\rnd)}(\mbf{x})=\kappa(\mbf{x})
\end{equation}
for any W-state $\ket{\mbf{x}}$ such that $x_0=0$ \cite{Chitambar-2012a}.  

Our goal now is to establish a similar result for concurrence distillation.
\begin{lemma}
\label{Lem:total-concurrence}
Let $\zeta(\mbf{x}):=2\sqrt{x_{n_1}x_{n_2}}+\frac{2}{3}x_{n_3}\sqrt{\frac{x_{n_1}}{x_{n_2}}}+\frac{1}{3}\frac{x_{n_2}x_{n_3}}{x_{n_1}}$, as introduced above.  Then
\begin{equation}
\label{Eq:C-random-lemma}
    C^{(\rnd)}(\mbf{x})=\zeta(\mbf{x})
\end{equation}
for any W-class state $\ket{\mbf{x}}$.  Furthermore, $\zeta(\mbf{x})$ is an entanglement monotone that strictly decreases whenever party $n_2$ or $n_3$ performs a non-trivial measurement.
\end{lemma}

\begin{remark}
The equality in Eq. \eqref{Eq:C-random-lemma} holds for all W-class states, even those for which $x_0\not=0$.  This is in contrast to Eq. \eqref{Eq:E2-random} which holds only under the assumption $x_0=0$.  The value $E^{(\rnd)}_2(\mbf{x})$ for general W-class states remains an open problem.
\end{remark}

\begin{proof}
The monotonicity of $\zeta(\mbf{x})$ and it strictly decreasing under measurement of party $n_2$ and $n_3$ is proven in Appendix \ref{Appendix-concurrence-monotone}.  Here we prove the achievability.  Consider an arbitrary W-class state $\ket{\mbf{x}}=\sqrt{x_0}\ket{000}+\sqrt{x_A}\ket{100}+\sqrt{x_B}\ket{010}+\sqrt{x_C}\ket{001}$.  Since we are trying to optimize the total average bipartite concurrence for any two parties, without loss of generality we can relabel the parties and assume that $x_A\geq x_B\geq x_C$. As depicted in Fig. \ref{Fig:total-concurrence}, the following three-step protocol achieves a total average bipartite concurrence of $\overline{C}=\zeta(\mbf{x})-\delta$ for arbitrary $\delta>0$.\\
\begin{figure}[t]
    \centering
    \includegraphics[width=8cm]{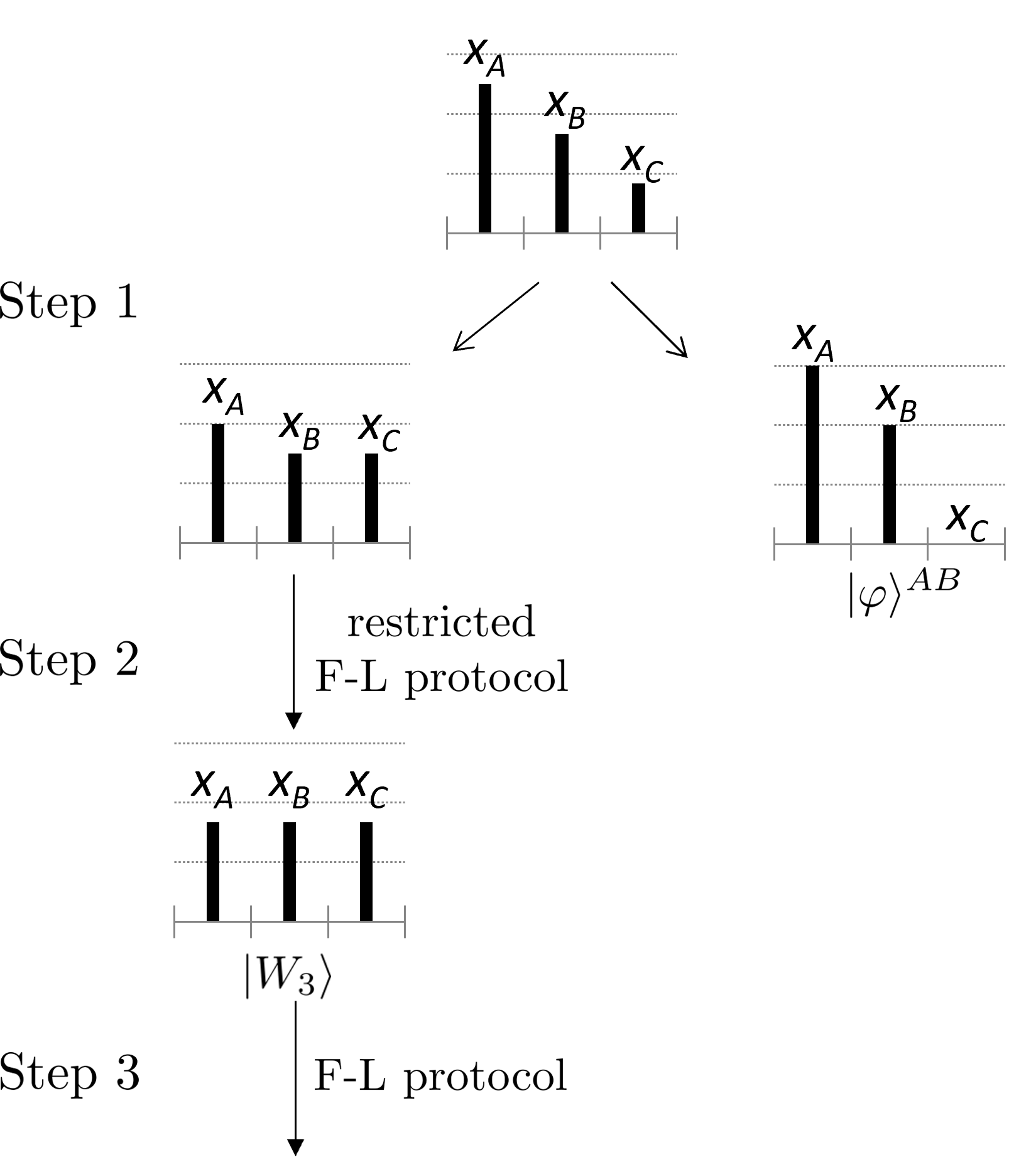}
    \caption{A depiction of the protocol described in Lemma \ref{Lem:total-concurrence}.}
    \label{Fig:total-concurrence}
\end{figure}

\noindent\textbf{Step 1.} Let Charlie perform a measurement described by the following Kraus Operators,\begin{align}
\label{Eq:Charlie-eq-vanish}
    &M_0=\sqrt{\frac{x_{C}}{x_{B}}}\op{0}{0}+\op{1}{1}, &M_1=\sqrt{1-\frac{x_{C}}{x_{B}}}\op{0}{0}.
\end{align}
There will be two measurement outcomes. Using the transformation rule of Eq. \eqref{Eq:coordinate-transformation-not-measuring}, the post-measurement state of outcome 1 is the bipartite pure state $\ket{\varphi}^{AB}=\frac{1}{\sqrt{p_1}}\sqrt{1-\tfrac{x_C}{x_B}}(\sqrt{x_0}\ket{00}+\sqrt{x_A}\ket{10}+\sqrt{x_B}\ket{01})$.  Hence from Eq. \eqref{Eq:bipartite-concurrence}, the concurrence of this state weighted by its probability is $C_1:=p_1C(\varphi)=2(1-\tfrac{x_C}{x_B})\sqrt{x_Ax_B}$.  The post-measurement state of outcome 0 is still a tripartite W-class state, which, when weighted by the probability of outcome $0$ has components $(\frac{x_Ax_C}{x_B}, x_C, x_C)$.

\noindent\textbf{Step 2.} Continuing with the post-measurement state of outcome 0, Bob and Charlie then perform $n$ iterations of a \textit{restricted} Fortescue-Lo protocol, which involves just Bob and Charlie performing $n$ iterations of the the weak measurement specified by Kraus operators in Eq. \eqref{F-L-Kraus}.  Here, $n$ is a parameter that we are free to choose, and for a given choice, we take $\epsilon>0$ such that $(1-\epsilon)^n=x_B/x_A$. So when $n\rightarrow \infty$ we have $\epsilon\rightarrow 0$.  

For each iteration of Bob and Charlie measuring $\{M_0^\epsilon,M_1^\epsilon\}$, they halt if at least one of them obtains outcome $1$.  If the both obtain outcome $0$ then they proceed with another iteration of measurement.  This is done for $n$ total iterations.  One can verify that if the protocol has not halted by the $j^{th}$ iterations, then outcome $00$ in the $j^{th}$ iteration yields the unnormalized state with coordinates $((1-\epsilon)^{2j} \frac{x_Ax_C}{x_B}, (1-\epsilon)^jx_C, (1-\epsilon)^j x_C)$.  On the other hand, if either outcome $01$ or $10$ is obtained, then the concurrence of the post-measurement state $\ket{\phi_{2,j}}$, weighted by its probability, is 
\begin{equation}
\label{Eq:Step-2-concurrence}
    C^j_2:=
    (1-\epsilon)^{(j-1)3/2}\cdot
    4\epsilon\sqrt{1-\epsilon}\cdot x_C\sqrt{\frac{x_A}{x_B}}.
\end{equation}
Should the protocol not halt for all $n$ iterations, we find that all all three coefficients end up being equal, $(1-\epsilon)^{2n} \frac{x_Ax_C}{x_B} = (1-\epsilon)^n x_C = \frac{x_Bx_C}{x_A}$, since $(1-\epsilon)^n=x_B/x_A$. After normalization, this state is the W-state $\ket{W}$.\\

\noindent\textbf{Step 3.}  The original F-L protocol is performed for $n'$ iterations on $\ket{W}$ with all three parties using measurement $\{M_0^{\epsilon'},M_1^{\epsilon'}\}$ and $\epsilon'$ a freely chosen parameter.  If a bipartite entangled state is obtained in the $j^{th}$ iteration of the F-L protocol, then the probability-weighted concurrence is
\begin{align}
     C_3^j=(1-\epsilon^\prime)^{2(j-1)}\cdot 6\epsilon^\prime(1-\epsilon^\prime) \frac{x_Bx_C}{x_A}.
\end{align}
As we let $n,n\to\infty$ and $\epsilon'\to 0$, the expected concurrence is given by
\begin{align}
   \overline{C}:=& C_1+\lim_{n\to\infty}\sum_{j=1}^nC_2^j+\lim_{\epsilon'\to 0}\lim_{n'\to\infty}\sum_{j=1}^{n'}C_3^j\notag\\
   =&2\sqrt{x_Ax_B}+\frac{2}{3}x_C\sqrt{\frac{x_A}{x_B}}+\frac{1}{3}\frac{x_Bx_C}{x_A}\nonumber \\
    =&\mathcal \zeta(\mbf{x}).
\end{align}

\end{proof}

Notice that this protocol uses two unbounded-round subroutines: the restricted F-L protocol (step 2) and the total F-L protocol (step 3).  In Section \ref{Sect:lower-bounds-F-L-Protocols} we will explore finite-round approximations and how the expected concurrence changes as we trade round numbers between these two subroutines.

\section{$\bigstar$-Random-Party Distillation}
\label{Sect:Star-Random-Party-Distillation}

The problem of $\star$-random-party EPR distillation in W-class states has been partially solved in Ref. \cite{Chitambar-2014b}.  Namely, it was shown that
\begin{align}
    E_2^{(\star\text{-}\rnd)}(\mbf{x})=\begin{cases} 2x_{\star}\quad\text{if} \quad x_{\star}\not=\max\{x_A,x_B,x_C\}\\
    2\eta(\mbf{x})\quad\text{if}\quad x_{\star}=\max\{x_A,x_B,x_C\}
    \end{cases}
    \label{Eq:E2-star-random}
\end{align}
for any W-state $\ket{\mbf{x}}$ such that $x_0=0$.  Similar to the scenario of total random-party distillation, the value of $ E_2^{(\star\text{-}\rnd)}(\mbf{x})$ is unknown when $x_0\not=0$.  However, the RHS of Eq. \eqref{Eq:E2-star-random} still serves as an upper bound due to the monotonicity of $x_\star$ and $\eta(\mbf{x})$.  Furthermore, the protocol achieving these distillation probabilities requires no more than three rounds of LOCC.

Our contribution here is solving the $\star$-Random-Party distillation problem for concurrence.
\begin{lemma}
\label{Lem:restricted-concurrence}
For arbitrary $\star\in\{A,B,C\}$, let  $\zeta^\star(\mbf{x})=2\sqrt{x_{\star}x_{n_1}}+\frac{2}{3}x_{n_2}\sqrt{\frac{x_{\star}}{x_{n_1}}}$ with $n_1,n_2\in\{A,B,C\}\setminus\{\star\}$ satisfying $x_{n_1}\geq x_{n_2}>0$.  Then
\begin{equation}
    C^{(\star\text{-}\rnd)}(\mbf{x})=\zeta^\star(\mbf{x}).
\end{equation}
\end{lemma}

\begin{proof}
In Ref. \cite{Chitambar-2014b} it was shown that $\zeta^\star(\mbf{x})$ is an entanglement monotone that that strictly decreases on average whenever party $\star$ or $n_1$ performs a non-trivial measurement.  What remains to be demonstrated is that $\zeta^\star(\mbf{x})$ is indeed an achievable concurrence distillation average for pairs $(\star,x_{n_1})$ and $(\star,x_{n_2})$.

Without loss of generality we assume $x_A>0$ since otherwise the lemma is trivially true.  The distillation protocol is very similar to the one given in Lemma \ref{Lem:total-concurrence}.  Formally, we replace $x_\star\leftrightarrow x_A$, $x_{n_1}\leftrightarrow x_B$, and $x_{n_2}\leftrightarrow x_C$.  However, while it was assumed that $x_A\geq x_B\geq x_C$ in Lemma \ref{Lem:total-concurrence}, we only assume that $x_B\geq x_C$ in this protocol.  

Step 1 of this protocol is then the same Step 1 as Lemma \ref{Lem:total-concurrence}, with Charlie performing the measurement given in Eq. \eqref{Eq:Charlie-eq-vanish}.  The weighted concurrence for outcome 1 is $C_1:=p_1C(\varphi)=2(1-\frac{x_C}{x_B})\sqrt{x_A x_B}$.  In step 2, Bob and Charlie again perform the restricted F-L protocol; however now there is no halting round.  The weighted concurrence of a bipartite state $\ket{\phi_{2,j}}$ obtained in round $j$ is $C_2^j$, as given in Eq. \eqref{Eq:Step-2-concurrence}.  Note that Alice then never performs a measurement in this protocol.  After replacing our system labels by $\{\star,n_1,n_2\}$, the total expected concurrence as the restricted F-L protocol in step 2 continues indefinitely is 
\begin{align}
    \overline{C}&=C_1+\lim_{n\to\infty}\sum_{j=1}^nC_2^n=2\sqrt{x_\star x_{n_1}}+\frac{2}{3}x_{n_2}\sqrt{\frac{x_\star }{x_{n_1}}}\notag\\
    &=\zeta^\star(\mbf{x}).
\end{align}
\end{proof}

\section{Round Complexity in Random-Party Distillation}

We now turn to the question of round complexity, which is the motivating topic of this work. 

\subsection{Lower Bounds in EPR Distillation}

\label{Sect:lower-bound-EPR}

Suppose the parties are only afforded a finite number of LOCC rounds.  What is the largest achievable probability of obtaining an EPR state starting from $\ket{W}$?  Here we prove a lower bound on this probability as a function of round number.

Consider any finite-round random EPR distillation protocol $\mc{P}$.  We can assume without loss of generality that every branch in the protocol ends with either an EPR state $\ket{\Psi^+}$ or a product state; this is because every weakly entangled state $\ket{\phi}$ can be transformed into $\ket{\Psi^+}$ with some probability without extending the number of rounds.  Any branch that terminates with an EPR state will be called a \textit{success branch} (many different success branches will overlap).

Our first step will be to argue that we can always transform the protocol $\mc{P}$ into one in which Alice and Bob just perform diagonal measurements.  As described in the introduction, we can characterize each local measurement by upper-triangular Kraus operators, $M_i=\left(\begin{smallmatrix}\sqrt{a_i}&b_i\\0&\sqrt{c_i}\end{smallmatrix}\right)$.  Notice that the composition of two such matrices has form
\begin{equation}
    \label{Eq:Kraus-composition}
M_iM_i'=\left(\begin{smallmatrix}\sqrt{a_ia_i'}&b_i\sqrt{c_i'}+b_i'\sqrt{a_i}\\0&\sqrt{c_ic_i'}\end{smallmatrix}\right).
\end{equation}
If we then consider the full set of Kraus operators constituting protocol $\mc{P}$, each success branch $\lambda$ will be characterized by a product Kraus operator of the form
\[T_{\lambda}=\left(\begin{smallmatrix}\sqrt{a_{1,\lambda}}&b_{1,\lambda}\\0&\sqrt{c_{1,\lambda}}\end{smallmatrix}\right)\otimes \left(\begin{smallmatrix}\sqrt{a_{2,\lambda}}&b_{2,\lambda}\\0&\sqrt{c_{2,\lambda}}\end{smallmatrix}\right)\otimes \left(\begin{smallmatrix}\sqrt{a_{3,\lambda}}&b_{3,\lambda}\\0&\sqrt{c_{3,\lambda}}\end{smallmatrix}\right),\]
where each matrix in the tensor product is obtained by concatenating all the local operators along branch $\lambda$ performed by Alice, Bob, and Charlie, respectively.  Since the only maximally entangled two-qubit state of the form $\ket{\mbf{x}}$ is $\ket{\Psi^+}$, we must have $T_\lambda\ket{W}\propto\ket{\Psi^+}\ket{0}$, with $\ket{\Psi^+}$ being held by some pair of parties.  The branch probability is then given by $|\bra{\Psi^+}\bra{0}T_\lambda\ket{W}|^2$, which is independent of the off-diagonal terms $b_{k,\lambda}$ in $T_\lambda$.  Furthermore,  by repeatedly applying the composition rule of Eq. \eqref{Eq:Kraus-composition}, the diagonal terms in $T_\lambda$ only depend on the diagonal terms of the individual local measurements in each round.  Hence, we would obtain the same success branch probabilites $|\bra{\Psi^+}\bra{0}T_\lambda\ket{W}|^2$ if Alice and Bob only performed the diagonal parts of their local Kraus operators.  In general, the diagonal parts of the Kraus operators will not form a complete measurement themselves, but they can easily be completed without affecting the overall success of the protocol.  Specifically, if the $M_i=\left(\begin{smallmatrix}\sqrt{a_i}&b_i\\0&\sqrt{c_i}\end{smallmatrix}\right)$ are the Kraus operators for some local measurement, then it is replaced by the Kraus operators $M'_i=\left(\begin{smallmatrix}\sqrt{a_i}&0\\0&\sqrt{c_i}\end{smallmatrix}\right)$ with an additional outcome $M'_0=\left(\begin{smallmatrix}0&0\\0&\sqrt{\sum_{i}|b_i|^2}\end{smallmatrix}\right)$.  Hence we have established the following simplification.
\begin{proposition}
\label{Prop:diagonal}
The optimal success probability for random-party EPR distillation of $\ket{W}$ in $N$ rounds can always be obtained by an $N$-round protocol in which all parties perform local measurements with diagonal Kraus operators.
\end{proposition}

With this proposition, let us consider an arbitrary distillation protocol $\mc{P}$ with all measurements in diagonal form.  As described in Section \ref{Sect:W-class-structure}, using weak measurement theory, we can depict the protocol as a tree that starts with the state $\ket{W}$ and evolves continuously along different branches.  Crucially, the conversion of a general protocol to one in which the coordinates $(x_A,x_B,x_C)$ transform smoothly can be done without changing the number of LOCC rounds.  This is because each local measurement is partitioned into a sequence of weak measurements with no communication needed between each measurement in the sequence.

To proceed with the analysis, we will need to introduce some concepts that characterize a given protocol $\mc{P}$.  We begin by defining an important class of states that appear in $\mc{P}$.
\begin{definition}
A state $\ket{\mbf{x}}$ with ordered components $x_{n_1}\geq x_{n_2}\geq x_{n_3}$ obtained in $\mc{P}$ is called a \textit{block state} if  (i) $x_{n_1}=x_{n_2}$ and $x_{n_3}\not=0$, (ii) a party whose component value is maximal is the next to perform a local measurement on $\ket{\mbf{x}}$, and (iii) an EPR state is obtained with some nonzero probability from $\ket{\mbf{x}}$.  
\end{definition}
\noindent Using the concept of a block state, we can analyze an arbitrary random-party distillation protocol in a systematic way.  First, with continuous evolution of the states in the protocol, it follows that the party having the largest component value cannot change along a success branch without first passing through a block state.  We can then partition the protocol in blocks determined by when there is a change in the party having the largest component value.  More precisely, define a sub-block of $\mc{P}$ as any sub-tree in the overall protocol tree that begins with a block state $\ket{\mbf{x}}$ and terminates as soon as it reaches some latter block state $\ket{\mbf{y}}$ (which may be equivalent to $\ket{\mbf{x}}$), an EPR state, or a product state (see Fig. \ref{Fig:sub-block-new}).  Each terminal block state in one sub-block is then the initial block state of a subsequent sub-block (unless the protocol halts), and the entire LOCC protocol can be divided into a disjoint union of sub-blocks.  We say that an $N$-block protocol is one that traverses $N$ different sub-blocks along the longest branch of the protocol.  We further organize the protocol into $N$ layers such that all sub-blocks belonging to layer $k$ have an initial block state laying on a branch that has previously traveled through $k-1$ sub-blocks.  Finally, we introduce the notion of a canonical sub-block.
\begin{definition}
A sub-block is called \textit{canonical} if all its terminal block states are $\ket{W}$ (see Fig. \ref{Fig:canonical-block} (a)).  Hence a canonical sub-block always ends with $\ket{W}$, an EPR state, or a product state.  A canonical sub-block is called \textit{$W$-canonical} if it begins with $\ket{W}$ and returns to $\ket{W}$ along only one branch within the sub-block (see Fig. \ref{Fig:canonical-block} (b)).
\end{definition}
\noindent Note that all sub-blocks in the $N^{th}$ layer are canonical since they obtain only EPR and product states.

\begin{figure}[t]
\centering
\includegraphics[scale=0.6]{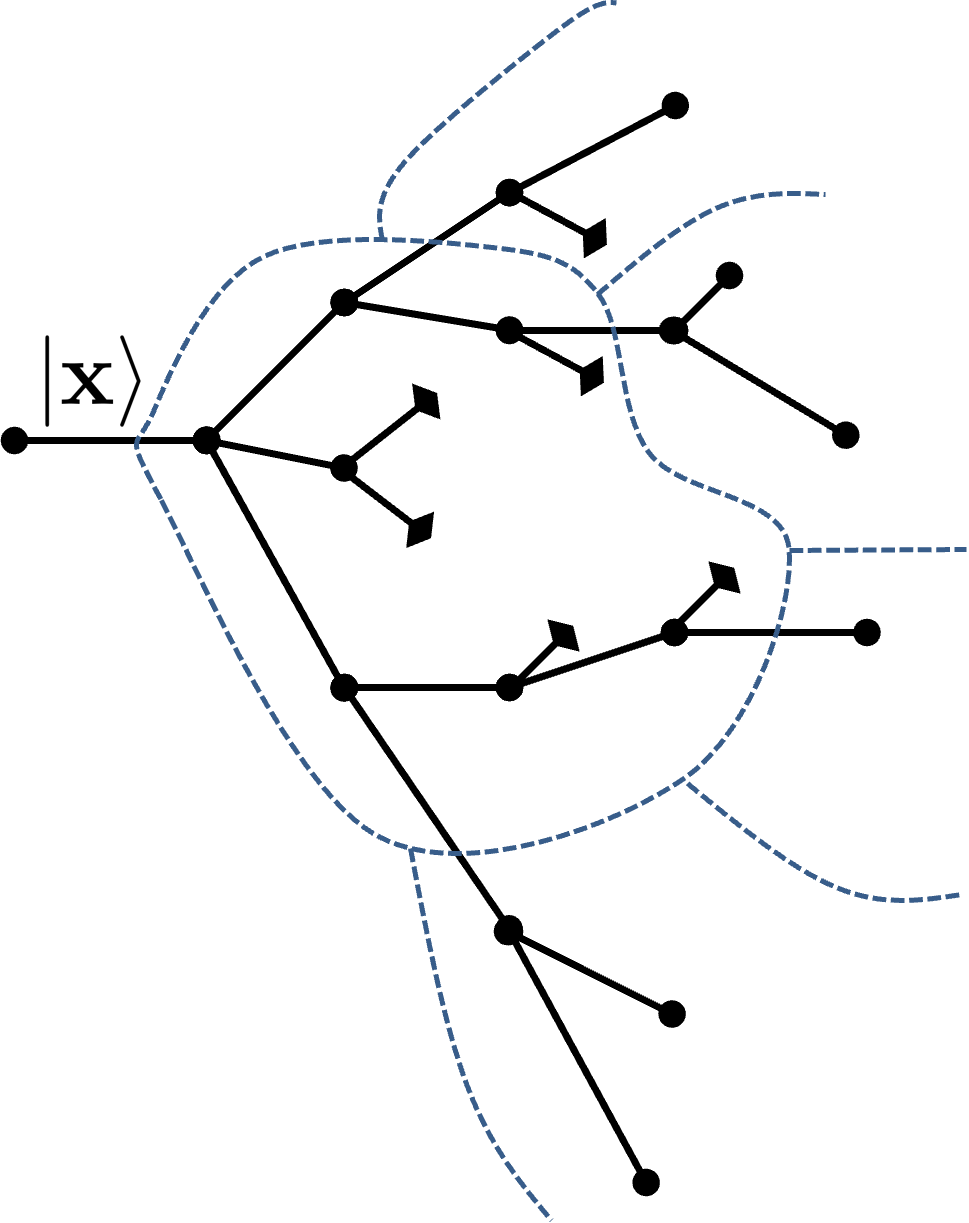}
\caption{ Any random-distillation protocol $\mc{P}$ can partitioned into disjoint sub-blocks (enclosed by dashed line).  Each sub-block starts with a block state and terminates with either another block state, an EPR state, or a product state.  EPR and product states are indicated by diamond end points.  Block states are obtained along any edge intersecting a dashed line.  A terminal block state of one sub-block is the initial block state of a subsequent sub-block unless the protocol  halts. }
\label{Fig:sub-block-new}
\end{figure}

Our first goal will be to reduce any $N$-block protocol $\mc{P}$ into another $\mc{P}'$ whose sub-blocks are $W$-canonical. The following technical lemma and its corollary provide the main ingredients for this procedure.
\begin{lemma}
\label{Lem:complexity}
Suppose that an arbitrary state $\ket{\mbf{x}}=(x_{n_1},x_{n_2},x_{n_3})$ is transformed into either $\ket{W}$, an EPR pair, or a product state with respective probabilities $P_W$, $P_{\text{EPR}}$, and $P_{F}$.  Further suppose that party $n_1(\mbf{x})$ maintains the largest component value in each of the final states (if more than one party has maximal value, fix $n_1(\mbf{x})$ to be one them).  Then \begin{align}
    P_{\text{EPR}}&\leq 2(x_{n_2}+x_{n_3})-4\frac{x_{n_2}x_{n_3}}{x_{n_1}}\label{Eq:Lemma-round-EPR-bound}\\
    P_W&=\frac{3}{2}(x_{n_2}+x_{n_3})-\frac{3}{4}P_{\text{EPR}}.\label{Eq:Lemma-round-W-equality}
\end{align}
Moreover, the upper bound on $P_{\text{EPR}}$ is achievable.
\end{lemma}
\begin{proof}
Since we are considering only protocols with diagonal local Kraus operators, the average component value for each party remains unchanged.  Hence, $x_{n_1}=\frac{1}{2}P_{\text{EPR}}+\frac{1}{3}P_{W}+P_F$.  Since $P_F=1-P_{\text{EPR}}-P_W$ and $x_{n_1}+x_{n_2}+x_{n_3}=1$, we have $x_{n_2}+x_{n_3}=\frac{1}{2}P_{\text{EPR}}+\frac{2}{3}P_W$, which is Eq. \eqref{Eq:Lemma-round-W-equality}.  On the other hand, the $\eta$ monotone says that $\eta(\mbf{x})=x_{n_2}+x_{n_3}-\frac{x_{n_2}x_{n_3}}{x_{n_1}}\geq \frac{1}{2}P_{\text{EPR}}+\frac{1}{3}P_{W}$.  Combining this with Eq. \eqref{Eq:Lemma-round-W-equality} yields Eq. \eqref{Eq:Lemma-round-EPR-bound}.

The upper bound on $P_{\text{EPR}}$ is achievable using an ``equal or vanish'' protocol \cite{Chitambar-2012b}.  This involves party $n_k(\mbf{x})$, for $k=2,3$, performing a binary-measurement with Kraus operators $M_1=\text{Diag}[\sqrt{x_{n_k}/x_{n_1}},1]$ and $M_2=\text{Diag}[\sqrt{1-x_{n_k}/x_{n_1}},0]$.  If they both obtain outcome $1$ the resulting state is $\ket{W}$, and this occurs with probability $P_W=3 \frac{x_{n_2}x_{n_3}}{x_{n_1}}$.  On the other hand, if one and only one party obtains outcome $1$, then an EPR pair is obtained.  This occurs with total probability 
\begin{align}
    P_{\text{EPR}}&=2 x_{n_2}(1-\frac{x_{n_3}}{x_{n_1}})+2 x_{n_3}(1-\frac{x_{n_2}}{x_{n_1}})\notag\\
    &=2(x_{n_2}+x_{n_3})-4\frac{x_{n_2}x_{n_3}}{x_{n_1}}.\notag
\end{align}
\end{proof}

\begin{figure}[t]
\centering
\includegraphics[scale=0.5]{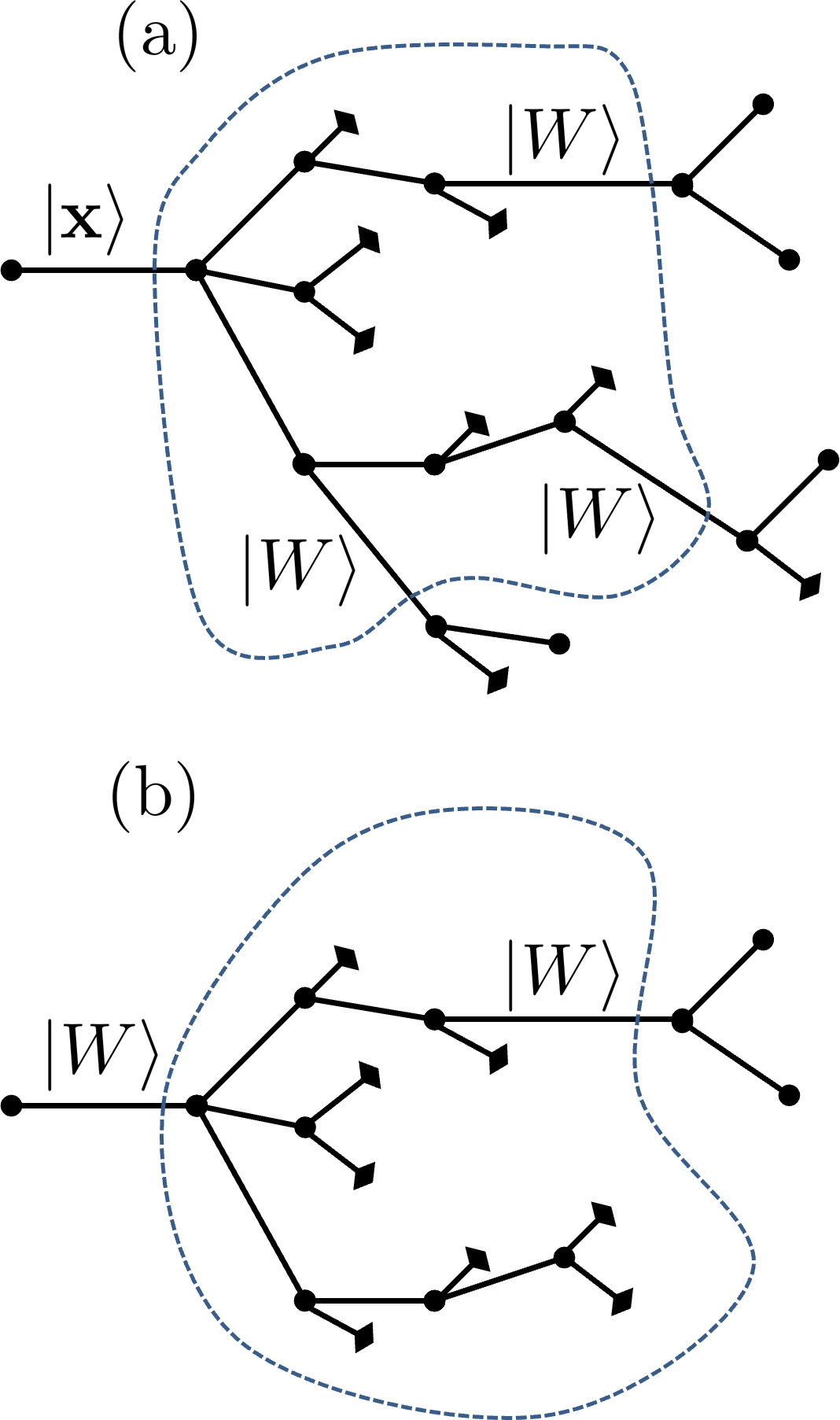}
\caption{(a) A canonical sub-block is one in which every terminating block state is $\ket{W}$.  (b) A $W$-canonical sub-block is a canonical sub-block whose initial block state is $\ket{W}$ and that has at most one terminating block state.}
\label{Fig:canonical-block}
\end{figure}

Lemma \ref{Lem:complexity} requires that party $n_1(\mbf{x})$ have the largest component value in all terminal states of the sub-block.  We next relax this condition and derive similar bounds on $P_{\text{EPR}}$ and $P_W$ when there is no party with maximal value in all the outcomes.  To conduct this analysis, let $a_i$ and $c_i$ be the numbers characterizing the measurement on an initial block state $\ket{\mbf{x}}$.  Let $\mathcal{A}$ (resp. $\mathcal{C}$) be the set of indices such that $a_i> c_i$ (resp. $c_i>a_i$).  Note that we can always assume that $a_i\not=c_i$ for every branch $i$.  Indeed, suppose that $a_i=c_i$, which means the state is left unchanged along branch $i$.  If all other branches lead to a higher expected EPR probability than the EPR probability along branch $i$, then one can just remove branch $i$ and scale the probability of the other branches up to unity.  On the other hand, if the subsequent sub-protocol along branch $i$ has a higher expected EPR probability than that along the other branches, one can just perform this sub-protocol on the original state $\ket{\mbf{x}}$.  Therefore, all branches with $a_i=c_i$ have been eliminated without decreasing the overall EPR probability or the number of rounds.  Since $\ket{\mbf{x}}$ is a block state, we have $x_{n_1}=x_{n_2}\geq x_{n_3}$.  The post-measurement states of $\ket{\mbf{x}}$ will thus split into two categories: for outcomes $i\in\mc{A}$ we will have $a_i x_{n_1}\geq a_i x_{n_3}>c_ix_{n_1}$ and for outcomes $i\in\mc{C}$ we will have $c_ix_{n_1}>a_ix_{n_1}\geq a_ix_{n_3}$.  In both these cases, the conditions of Lemma \ref{Lem:complexity} apply since the party having maximal component cannot switch without passing through another block state, and by assumption, the only block state obtained within the given sub-block is $\ket{W}$ (for which all parties have maximal component value).  We can then prove the following corollary.  
\begin{corollary}
\label{Cor:main}
Suppose that $\ket{\mbf{x}}$ is the initial block state of a canonical sub-block and a local measurement is performed with Kraus operator parameters $a_i$ and $c_i$.   Let $a=\sum_{i\in\mathcal{A}}a_i$, $c=\sum_{i\in\mathcal{A}}c_i$, and $\tilde{a}=a-c$. Then
\begin{align}
    P_{\text{EPR}}&\leq 2x_{n_1}(1-\tilde{a})+2x_{n_3}-4x_{n_3}(1-\tilde{a})^2\label{Eq:EPR-bound}\\
     P_W&=\frac{3}{2}x_{n_1}(1-\tilde{a})+\frac{3}{2}x_{n_3}-\frac{3}{4}P_{\text{EPR}}.
\end{align}
Furthermore, the upper bound of $P_{\text{EPR}}$ can be saturated using binary-outcome measurements that obtain $\ket{W}$ along just one branch. 
\end{corollary}
\begin{proof}
Group outcomes $i$ into sets $\mc{A}$ and $\mc{C}$ as described above.  Applying Lemma \ref{Lem:complexity} to each of the post-measurement states yields
\begin{align}
\label{Eq:UBPtot}
P_{\text{EPR}}\leq& \sum_{i\in\mathcal{A}}\left(2x_{n_1}c_i+2x_{n_3}a_i-4x_{n_3}c_i\right)\notag\\&+\sum_{i\in\mathcal{C}}\left(2x_{n_1}a_i+2x_{n_3}a_i-4x_{n_3}\frac{a_i^2}{c_i}\right)\notag\\
\leq& 2x_{n_1}(1-\tilde{a})+2x_{n_3}-4x_{n_3}\left(c+\frac{(1-a)^2}{1-c}\right)\notag\\
\leq& 2x_{n_1}(1-\tilde{a})+2x_{n_3}-4x_{n_3}(1-(a-c))^2\notag\\
=&2x_{n_1}(1-\tilde{a})+2x_{n_3}-4x_{n_3}(1-\tilde{a})^2,
\end{align}
where the second inequality follows from the Cauchy-Schwarz Inequality: 
\[\left(\sum_{i\in\mathcal{C}}\sqrt{c_i}\cdot\sqrt{c_i}\right)\left(\sum_{i\in\mathcal{C}}\frac{a_i}{\sqrt{c_i}}\cdot\frac{ a_i}{\sqrt{c_i}}\right)\geq(\sum_{i\in\mathcal{C}}a_i)^2,\] and the third comes from the fact that $c+\frac{(1-a)^2}{1-c}\geq(1-(a-c))^2$.  The equality of $P_W$ again follows from Lemma \ref{Lem:complexity} as
\begin{align}
    P_W&=\frac{3}{2}x_{n_1}\left(\sum_{i\in\mc{A}}c_i+\sum_{i\in\mc{C}}a_i\right)+\frac{3}{2}x_{n_3}-\frac{3}{4}P_{\text{EPR}}\notag\\
    &=\frac{3}{2}x_{n_1}(1-\tilde{a})+\frac{3}{2}x_{n_3}-\frac{3}{4}P_{\text{EPR}}.
\end{align}
The upper bound on $P_{\text{EPR}}$ can be achieved by having party $n_1(\mbf{x})$ perform a binary outcome measurement with Kraus operators $M_1=\text{Diag}[\sqrt{1-\tilde{a}},1]$ and $M_2=\text{Diag}[\sqrt{\tilde{a}},0]$.  For outcome $2$, an EPR pair can be subsequently obtained between parties $n_2(\mbf{x})$ and $n_3(\mbf{x})$ with probability $2\tilde{a}x_{n_3}$.  For outcome $1$, the ``equal or vanish'' protocol is subsequently performed, as described in Lemma \ref{Lem:complexity}.  Since $x_{n_1}=x_{n_2}$, the total EPR probbility is then
\begin{align}
    P_{\text{EPR}}=&2\tilde{a}x_{n_3}+2(1-\tilde{a})(x_{n_1}+x_{n_3})-4x_{n_3}(1-\tilde{a})^2.\notag
\end{align}
  From the ``equal or vanish'' protocol we also obtain a W-state with probability  $P_W=3x_{n_3}(1-\tilde{a})^2$.
\end{proof}

We next convert an arbitrary canonical sub-block into a $W$-canonical sub-block.
\begin{lemma}
\label{Lem:W-canonical}
Any canonical sub-block of $\mc{P}$ can be replaced by a $W$-canonical sub-block without decreasing the total EPR probability or increasing the total number of sub-blocks in the protocol.
\end{lemma}

\begin{figure*}[t]
\includegraphics[scale=0.7]{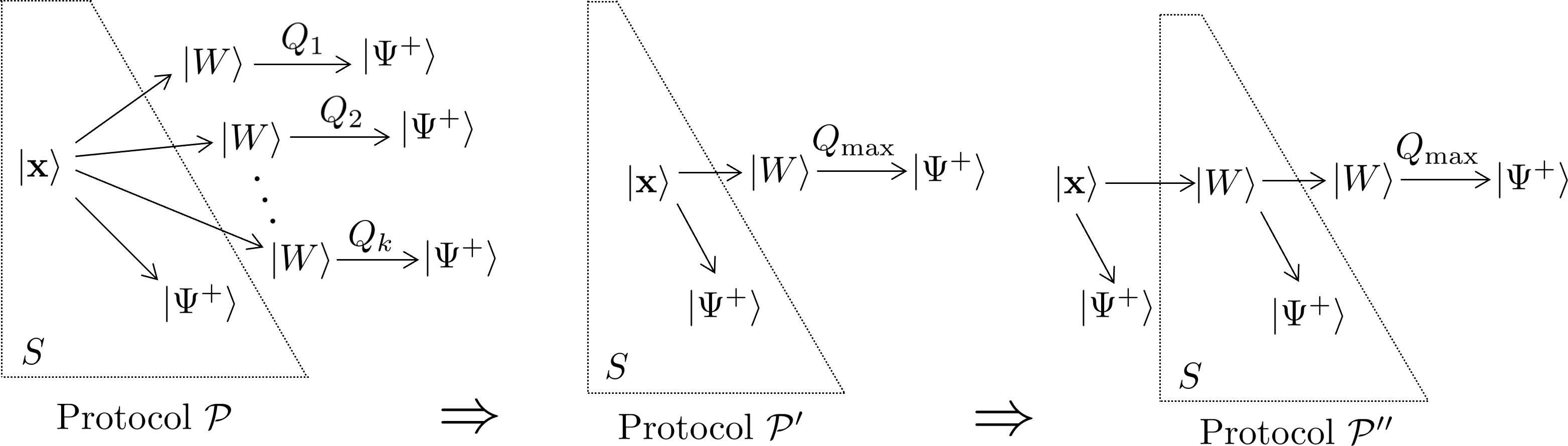}
\caption{\label{fig:protocol-reduction} Lemma \ref{Lem:W-canonical} describes the reduction of a canonical sub-block $S$ to a $W$-canonical sub-block in two steps.  The first is to reduce the number of terminal $\ket{W}$ states to just one.  The second replaces the initial block state $\ket{\mbf{x}}$ with $\ket{W}$ by performing a prior measurement on $\ket{\mbf{x}}$ in the preceding sub-block.}
\end{figure*}

\begin{proof}
Let $S$ be any canonical sub-block starting with block state $\ket{\mbf{x}}$.  Let $a_i$ and $c_i$ be the Kraus operator parameters of the initial measurement as before.  We transition $S$ into a $W$-canonical sub-block in two steps.  First, we limit the the total number of branches within the sub-block that obtain $\ket{W}$ to at most one.  In general, there may be multiple branches within the sub-block that obtain state $\ket{W}$.  For each of these, the protocol $\mc{P}$ specifies a subsequent distillation sub-protocol to be performed on $\ket{W}$.  Let $Q_{\max}$ denote the maximum probability of obtaining EPR pairs from $\ket{W}$ among all these sub-protocols.  Therefore, if $P_S$ is the total probability of obtaining EPR pairs when starting from the block state $\ket{\mbf{x}}$, we can use Corollary \ref{Cor:main} to obtain the bound
\begin{align}
P_S&\leq P_{\text{EPR}}+P_W\cdot Q_{\max}\notag\\
    &= P_{\text{EPR}}+\left(\frac{3}{2}x_{n_1}(1-\tilde{a})+\frac{3}{2}x_{n_3}-\frac{3}{4}P_{\text{EPR}}\right)Q_{\max}\notag\\
    &= \frac{3}{2}x_{n_1}(1-\tilde{a})+\frac{3}{2}x_{n_3}+\left(1-\tfrac{3}{4}Q_{\max}\right)P_{\text{EPR}}.
\end{align}
Since the coefficient of $P_{\text{EPR}}$ is non-negative on the RHS, we can further bound $P_S$ using Eq. \eqref{Eq:EPR-bound}.  This bound can be saturated by performing the optimal binary-outcome measurement scheme of Corollary \ref{Cor:main} on the initial block state $\ket{\mbf{x}}$.  For the single $\ket{W}$ outcome, the original sub-protocol that attains EPR probability $Q_{\max}$ is then performed.  Hence, if $P_S'$ denotes the total EPR probability starting from $\ket{\mbf{x}}$ in our modified protocol $\mc{P}'$, then
\begin{align}
    P_{S}'&=2x_{n_1}(1-\tilde{a})+2x_{n_3}-4x_{n_3}(1-\tilde{a})^2\notag\\&+3x_{n_3}(1-\tilde{a})^2\cdot Q_{\max}\notag\\
    &\geq P_S.
\end{align}

Next, we modify the initial block state $\ket{\mbf{x}}$ to be $\ket{W}$ (if $\ket{\mbf{x}}\not=\ket{W}$).  This is accomplished by stochastically transforming $\ket{\mbf{x}}$ into $\ket{W}$ or $\ket{\Psi^+}$ in the sub-block prior to $S$ by having party $n_3(\mbf{x})$ measure with Kraus operators $M_1=\text{Diag}[\sqrt{x_{n_3}/x_{n_1}},1]$ and $M_2=\text{Diag}[\sqrt{1-x_{n_3}/x_{n_1}},0]$.  The probability of obtaining $\ket{W}$ (outcome $1$) is $3x_{n_3}$ while the probability of $\ket{\Psi^+}$ (outcome $2$) is $2(x_{n_1}-x_{n_3})$.  Now $\ket{W}$ is the initial block state of $S$.  The optimal transformation of Corollary \ref{Cor:main} is performed on $\ket{W}$ with the measurement parameter $\tilde{a}$ being the same as in $\mc{P}'$.  This yields EPR states within sub-block $S$ with probability $P_{\text{EPR}}=\frac{2}{3}(1-\tilde{a})+\frac{2}{3}-\frac{4}{3}(1-\tilde{a})^2=\frac{2}{3}(3-2\tilde{a})\tilde{a}$, while the W-state is obtained with probability $P_W=(1-\tilde{a})^2$.  For the $\ket{W}$ outcome, the $Q_{\max}$ sub-protocol is again performed in subsequent sub-blocks.  Letting $P_{S}''$ denote the total EPR probability starting from $\ket{\mbf{x}}$ in this new protocol $\mc{P}''$, we have
\begin{align}
    P_S''=&2(x_{n_1}-x_{n_3})+2x_{n_3}(3-2\tilde{a})\tilde{a}\notag\\
    &+3x_{n_3}(1-\tilde{a})^2\cdot Q_{\max}.
\end{align}
One can verify that
\begin{align}
    P_S''-P_S\geq P_{S}''-P_S'=2a(x_{n_1}-x_{n_3})\geq 0.
\end{align}
Hence, we have transformed an arbitrary canonical sub-block $S$ into a $W$-canonical sub-block without decreasing the total EPR probability or increasing the number of sub-blocks in the protocol.

\end{proof}

\begin{corollary}
\label{Cor:main2}
Any $N$-block LOCC protocol $\mathcal{P}$ performing random-party EPR distillation can be converted into an $N$-block protocol consisting only of $W$-canonical sub-blocks and obtaining EPR pairs with at least as large as a probability as $\mathcal{P}$.
\end{corollary}

\begin{proof}
  All sub-blocks at layer $N$ are canonical, and we perform Lemma \ref{Lem:W-canonical} to convert them to $W$-canonical sub-blocks.  We then move backward to layer $N-1$.  All of these sub-blocks must be canonical since their terminal block states have been converted to $\ket{W}$ in the previous step.  We can thus transform all these canonical sub-blocks to $W$-canonical sub-blocks.  Proceeding iteratively layer-by-layer back to the original state $\ket{W}$, we obtain the stated structure of Corollary \ref{Cor:main2}.
\end{proof}

We now possess the necessary components to prove a general lower bound.
\begin{theorem}
\label{Thm:main}
Let $P_{\tot}$ be the total probability of obtaining EPR pairs from $\ket{W}$ in some LOCC random-party distillation protocol $\mc{P}$.  Then
\begin{equation}
\label{Eq:lower-bound}
\text{Rounds in $\mc{P}$}\geq \frac{1}{1-P_{tot}}-2.
\end{equation}
Moreover, this bound is tight.
\end{theorem}

\begin{figure}[b]
\centering
\includegraphics[scale=0.55]{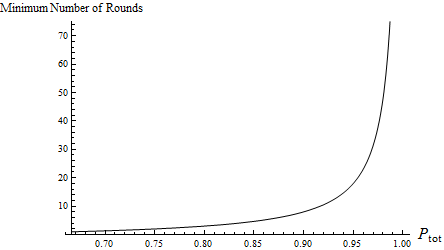}
\caption{\label{Fig:LowerBoundPlot} Minimum number of LOCC rounds to obtain EPR pairs from $\ket{W}$ with probability $P_{\tot}$.  This lower bound is tight.}
\end{figure}
\begin{proof}
Suppose that $\mc{P}$ is a finite-round protocol; otherwise Eq. \eqref{Eq:lower-bound} holds trivially.  In light of Corollary \ref{Cor:main2}, we can assume without loss of generality that $\mc{P}$ consists only of $W$-canonical sub-blocks.  For a non-negative number $a_k$, the probability of distilling EPR pairs in the $k^{th}$ sub-block is $\frac{2}{3}(2-a_k)-\frac{4}{3}B_{k}$, where $B_{k}=(1-a_k)^2$ is the probability of obtaining $\ket{W}$ and therefore carrying the protocol onto sub-block $k+1$.  Consequently, the total EPR probability for an $N$-block protocol is
\begin{align}
P_{\tot}=&\frac{4}{3}-\frac{2}{3}a_1-\frac{4}{3}B_{1}+B_{1}(\frac{4}{3}-\frac{2}{3}a_2-\frac{4}{3}B_{2})\notag\\&+...+\frac{2}{3}(B_1B_2\cdots B_{N-1})\notag\\
=&\frac{4}{3}-\frac{2}{3}a_1-\frac{2}{3}B_1a_2\notag\\&-...-\frac{2}{3}(B_1\cdots B_{N-2})a_{N-1}-\frac{2}{3}(B_1\cdots B_{N-1})\notag\\
=&\frac{4}{3}-\frac{2}{3}a_1-\frac{2}{3}\sum_{i=1}^{N-1}\prod_{j=1}^{i}B_ja_{i+1},
\end{align}
where $a_N=1$.  To optimize, we compute partial derivatives, which ultimately leads to a system of equations:
\begin{align}
0&=1-2(1-a_{N-1})\notag\\
0&=1-2(1-a_{N-2})(a_{N-1}+B_{N-1})\notag\\
0&=1-2(1-a_{N-3})(a_{N-2}+B_{N-2}(a_{N-1}+B_{N-1}))\notag\\
&\vdots\notag
\end{align}
The solutions are given by the recursive formula $a_k=\frac{a_{k+1}}{1+a_{k+1}}$ for $k=1,...,N-1$.  Thus, $a_k=\frac{1}{N-(k-1)}$ and $B_k=\left(\frac{N-k}{N-(k-1)}\right)^2$.  Substituting back into $P_{tot}$, we find that for any $N$-block LOCC distillation of $\ket{W}$,
\begin{align}
P_{\tot}&\leq\frac{4}{3}-\frac{2}{3}\frac{1}{N}-\frac{2}{3}\sum_{i=1}^{N-1}\prod_{j=1}^i\left(\frac{N-j}{N-(j-1)}\right)^2\left(\frac{1}{N-i}\right)\notag\\
&=\frac{4}{3}-\frac{2}{3}\frac{1}{N}-\frac{2}{3}\sum_{i=1}^{N-1}\frac{N-i}{N^2}\notag\\
&=1-\frac{1}{3N}.
\end{align}
This upper bound can be achieved by using measurements with the optimal values of $a_k$.  Finally, in terms of LOCC round number, we have that each sub-block contains at least three rounds unless it is the final sub-block, which only requires one sub-block since it involves fixed-party distillation.  Therefore, we obtain Eq. \eqref{Eq:lower-bound} with a corresponding plot in Fig. \ref{Fig:LowerBoundPlot}.
\end{proof}

\begin{remark}
Throughout this section we have been modeling an LOCC protocol as a discrete sequence time steps in which only a single party measures at each step.  In this model, the round number of the protocol is the largest number of times the measuring party switches in some branch.  By only allowing one party to measure at a time, the protocol is easier to analyze.  However, a more general model allows all parties to measure at each time step.  This has been referred to as \textit{broadcast} LOCC (BLOCC) in Ref. \cite{Gonzales-2020a}, and this model has also been consider in Ref. \cite{Rozpiedek-2018a}.  In the BLOCC setting, a protocol's round number is then just the largest number of broadcasts made along some branch.  In Section \ref{Sect:lower-bounds-F-L-Protocols}, we adopt the BLOCC model to count the number of iterations in the F-L protocol for convenience.\par Notice that the BLOCC round complexity of a task can be reduced by no more than a factor of $N$, compared to the ``one at a time'' round complexity.  In terms of random-party EPR distillation of $\ket{W}$, the lower bound of Theorem \ref{Thm:main} becomes
\begin{equation}
\notag
\text{Number of classical broadcasts in $\mc{P}$}\geq \frac{1}{3(1-P_{tot})}.
\end{equation}
This is also tight since each $W$-canonical sub-block requires only one round of broadcast communication to implement. 
\end{remark}

\begin{remark}
The lower bound of Theorem \ref{Thm:main} matches in form the original F-L distillation probability \cite{Fortescue-2007a} as well as the slightly optimized distillation probability by Li \textit{et al.} \cite{Li-2020a}.
\end{remark}

\subsection{Lower Bounds in Concurrence Distillation}

\subsubsection{Infinite Rounds is Optimal}

We now turn to the task of random-party concurrence distillation.  We first establish that every finite-round protocol is sub-optimal.
\begin{theorem}
\label{Thm:unbounded-concurrence}
Suppose $\mc{P}$ is an $N$-round total random-party concurrence distillation protocol on a tripartite W-state $\ket{\mbf{x}}$.  Then there exists a protocol $\mc{P}'$ having more than $N$ rounds but distills a larger average concurrence than $\mc{P}$.  The same statement holds for $\star$-random-party concurrence distillation.   
\end{theorem}
\begin{proof}
Consider the last instance in the protocol for which some local measurement is performed on a tripartite entangled W-state.  Call this state $\ket{\mbf{x}'}$.  Since this is the final time such a state appears in the protocol, the local measurement must be a hard measurement that completely decouples the measuring party from the other two.  Moreover, the protocol halts after this measurement since any further bipartite processing can never increase the concurrence.  The greatest average post-measurement concurrence obtained from a hard measurement on $\ket{\mbf{x}'}$ is called the concurrence of assistance (CoA) \cite{Laustsen-2003a}, and if $(x_A',x_B',x_C')$ are the coordinates of state $\ket{\mbf{x}'}$, then
\begin{equation}
\text{CoA of $\ket{\mbf{x}'}$} \leq 2\sqrt{x'_{i}x'_j},
\end{equation}
where $i,j\in\{A,B,C\}$ are the two non-measuring parties \cite{Laustsen-2003a}.  It is obvious that both $\zeta(\mbf{x}')$ and $\zeta^\star(\mbf{x}')$ are strictly larger than this value.  Hence if instead of performing the hard measurement and terminating the protocol the parties were to continue with the protocols of Lemma \ref{Lem:total-concurrence} and \ref{Lem:restricted-concurrence}, respectively, they would obtain a greater average concurrence.
\end{proof}
This result shows that infinite-round LOCC (i.e. protocols with unbounded round number) is needed to optimally distill concurrence from W-class states.

\begin{remark}
As noted above Lemma \ref{Lem:restricted-concurrence}, for the task of $\star$-random-party EPR distillation, an optimal protocol can be found that consumes only three rounds of LOCC.  Thus, we see that finite rounds of LOCC suffice to optimize the entanglement measure $E_{2}^{(\star-\rnd)}$, whereas infinite rounds of LOCC are needed to optimize the measure $C^{(\star-\rnd)}$, even though the tasks share the same goal of breaking tripartite entanglement into bipartite form.
\end{remark}

\begin{remark}
The proof of Theorem \ref{Thm:unbounded-concurrence} uses the fact that both $\zeta$ and $\zeta^\star$ are defined on the entire class of W-states, including those states $(x_A,x_B,x_C)$ for which $x_0\not=0$.  In contrast, the proof will not directly work for $E_2^{(\rnd)}$ since  the equality $E_2^{(\rnd)}(\x)=\kappa(\x)$ only holds for states with $x_0=0$, and a general LOCC protocol might attain states for which the latter condition does not hold.  However, by using Proposition \ref{Prop:diagonal} and its generalization to arbitrary W-class states, one can assume without loss of generality that only diagonal Kraus operators are used on the protocol.  Hence if the initial states $\x$ has $x_0=0$, then it will remain zero throughout.  Consequently, we can repeat the same argument as Theorem \ref{Thm:unbounded-concurrence} and likewise conclude that $E_2^{(\rnd)}(\x)$ is not achievable in finite round LOCC whenever $x_0=0$.
\end{remark}

\subsubsection{Lower Bounds for F-L Protocols}

\label{Sect:lower-bounds-F-L-Protocols}

We now consider the problem of finite round concurrence distillation. Specifically, we restrict the protocol to be the ones in Section \ref{Sect:Total-Random-Party-Distillation} and \ref{Sect:Star-Random-Party-Distillation}, but allow the $\epsilon$ parameters in the F-L measurements to vary depending on the round number. We find the optimal choices of these $\epsilon$'s given the total number of BLOCC rounds.\par

\noindent\textit{$\star$-Random-Party Concurrence Distillation.}
We first find the optimal protocol for the $\star$-random-party concurrence distillation task.  We consider just the type of protocol specified in Section \ref{Sect:Star-Random-Party-Distillation}.\par
The first step is always a hard measurement on party $n_2$ if $x_{n_2}\not=x_{n_1}$, but can be omitted if they are equal. We thus exclude this step in our round counts for either case, assuming an additional round if $x_{n_2}\not=x_{n_1}$. The rest of the protocol consists of Fortescue-Lo measurements but with the $\epsilon_k$ now able to vary for each round $k=1,\dots,N-1$. In the last round, we perform a hard measurement on party $n_2$ to retrieve the concurrence in the bipartite state between party $\star$ and $n_1$. Then the concurrence distilled for round $k$ is
\begin{align}
    C_k =
    \begin{cases}
                4x_{n_2}\sqrt{\frac{x_\star}{x_{n_1}}}\prod_{j=1}^{k-1}(1-\epsilon_j)^{\frac{3}{2}}
                \epsilon_k\sqrt{1-\epsilon_k}, &k<N,\\
                2x_{n_2}\sqrt{\frac{x_\star}{x_{n_1}}}\prod_{j=1}^{N-1}(1-\epsilon_j)^{3/2}, & k=N.
    \end{cases}
\end{align}
Our goal is to maximize the total concurrence $\sum_{k=1}^N C_k$ by varying $\epsilon_k$, $k=1,\dots,N$. By requiring the partial derivatives to be zero, the $\epsilon_k$'s are calculated to be
\begin{align}
    \epsilon_k =
    \begin{cases}
                1/3 & \text{for } k=N-1,\\
                 \frac{1}{1+1/(2-3A_k-\frac{3}{2}B_k)}& \text{for } k=1,\dots,N-2,
    \end{cases}
\end{align}
where\begin{align}
    A_k &= \sum_{j=k+1}^{N-1}\left[\prod_{i=k+1}^{j-1}(1-\epsilon_i)^{3/2}\right]\epsilon_j\sqrt{1-\epsilon_j},\\
    B_k &= \prod_{i=k+1}^{N-1}(1-\epsilon_i)^{3/2}.
\end{align}
Note that each $\epsilon_k$ only depends on $\epsilon_j$'s with $j>k$. So these parameters can be calculated one by one in a descending manner, starting from $\epsilon_{N-1}=1/3$. \par
For a $\ket{W}$ state, the relation between the number of rounds and the optimal average concurrence is plotted in Fig. \ref{Fig:StarLowerBoundPlot}. Note that as the number of rounds increases, the concurrence distilled approaches the asymptotic bound $\zeta^\star$.
\begin{figure}[H]
\includegraphics[width=8cm]{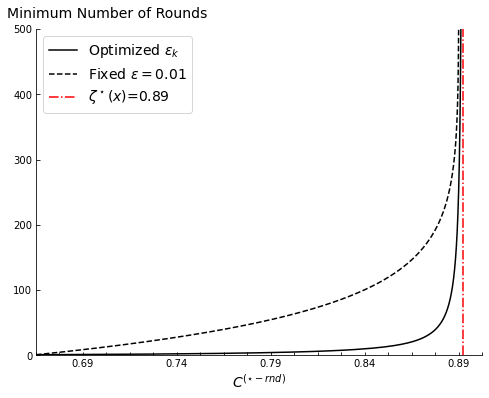}
\caption{\label{Fig:StarLowerBoundPlot} Minimum number of BLOCC rounds to obtain average bipartite concurrence $C^{(\star\text{-}\rnd)}$ from $\ket{W}$. Comparison is made between choosing a uniform $\epsilon$ and optimizing $\epsilon_k$ for each round.}
\end{figure}

\noindent\textit{Total Random-Party Concurrence Distillation.}
\begin{figure}[h]
\includegraphics[width=8cm]{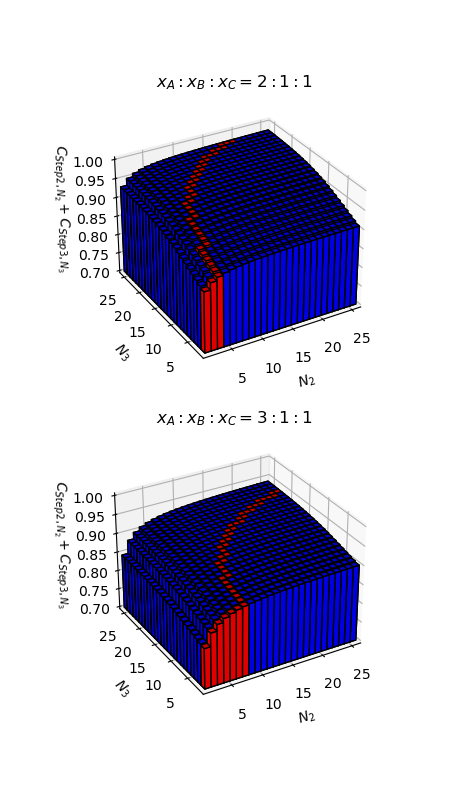}
\caption{\label{Fig:3DPlot} Maximum average bipartite concurrence that can be distilled using $N_2$ number of BLOCC rounds in Step 2 and $N_3$ in Step 3. Each red bar represents the best value in the diagonal line where $N_2+N_3$ is constant. Notice that for different states there are different ways to optimally distribute the number of rounds into each step. For states with a large ratio between $x_A$ and $x_C$, more rounds are needed in Step 2 to gradually lower the ratio.}
\end{figure}
We next find the optimal finite-round parameters for total random party concurrence distillation protocol described in Section \ref{Sect:Total-Random-Party-Distillation}. Notice that in this case there are two separate steps (Step 2 and 3) that require infinitesimal measurements to achieve the asymptotic bound. Therefore, for a fixed finite number of rounds $N$, we distribute it so that we spend $N_2$ in step 2 and $N_3$ in step 3 with $N=N_2+N_3$.\par
In step 2, the expected concurrence distilled for round $k$ is 
\begin{align}
    C_{\text{Step }2,k} = 4x_{C}\sqrt{\frac{x_A}{x_{B}}}\prod_{j=1}^{k-1}(1-\epsilon_j)^{3/2}\epsilon_k\sqrt{1-\epsilon_k}, 
\end{align}
where $k=1,\dots,N_2$. As in the asymptotic protocol, we require that $\prod_{k=1}^{N_2}(1-\epsilon_k)=x_B/x_A=\gamma$. Given a fixed $N_2$, we wish to maximize $C_{\text{Step }2, N_2}:=\sum_{k=1}^{N_2} C_{\text{Step }2,k}$ over $\epsilon_k$'s subject to the constraint above.  Using the Lagrange multiplier method, the maximizing condition is 
\begin{align}
    \nabla_{\lambda, \epsilon_1, \dots ,\epsilon_{N_2}} \left[ \sum_{k=1}^{N_2} C_{\text{Step }2,k} + \lambda(\prod_{k=1}^{N_2}(1-\epsilon_k)^{3/2}-\gamma^{3/2})\right]&\nonumber\\
    =0.&
\end{align}
Solving this we get \begin{align}
    (1-\epsilon_{k})=\begin{cases}
                \frac{1}{3(1-\lambda)}, & \text{for } k=N_2,\\
                \frac{1}{3-2\sqrt{1-\epsilon_{k+1}}}, & \text{for } k=1, 2,\dots, N_2-1.
    \end{cases}
\end{align}
Therefore, every $\epsilon_k$ can be expressed in terms of $\lambda$, and the exact number can be calculated by plugging into the constraint that $\prod_{k=1}^{N_2}(1-\epsilon_k)=x_B/x_A=\gamma$.\par
In step 3, the expected concurrence distilled for round $k$ is 
\begin{align}\label{Eqn:concurrence-step3}
    C_{\text{Step }3,k} = \begin{cases}
               6\frac{x_Bx_C}{x_A}\prod_{j=1}^{k-1}(1-\epsilon_j)^{2}
                \epsilon_k(1-\epsilon_k),  &k<N_2,\\
                2\frac{x_Bx_C}{x_A}\prod_{j=1}^{k-1}(1-\epsilon_j)^{2},  &k=N_2,
    \end{cases}
\end{align}
with a hard measurement on party C in the end. (With an abuse of notation we use $\epsilon_k$'s for the parameters in both step 2 and step 3, but they are different variables.) With a fixed $N_3$, we wish to maximize $C_{\text{Step }3,N_3}:=\sum_{k=1}^{N_3}C_{\text{Step }3,k}$ over $\epsilon_k$'s. By requiring the partial derivatives with respect to $\epsilon_k$'s to be zero, we get 
\begin{align}
    \epsilon_k =
    \begin{cases}
                1/4 & \text{for } k=N_3-1,\\
                 1-\frac{3}{2(3-3A_k-B_k)}& \text{for } k=1,\dots,N_3-2,
    \end{cases}
\end{align}
where\begin{align}
    A_k &= \sum_{j=k+1}^{N_3-1}\left[\prod_{i=k+1}^{j-1}(1-\epsilon_i)^{2}\right]\epsilon_j\sqrt{1-\epsilon_j},\\
    B_k &= \prod_{i=k+1}^{N_3-1}(1-\epsilon_i)^{2}.
\end{align}
Again, the $\epsilon_k$'s can be calculated from the last one to the first, and the total concurrence distilled can be calculated by plugging the values into \eqref{Eqn:concurrence-step3}. 

The final goal is to determine the best way to distribute the total number of rounds $N$ into $N_2$ and $N_3$, i.e. \begin{align}
    \max_{N_2, N_3\in \mathbb N \text{ s.t. } N=N_2+N_3} C_{\text{Step }2,N_2}+C_{\text{Step }3,N_3}.
\end{align} This can be done by brute-force search over the $N-1$ ways of distributing round numbers to Step 2 and 3. Results are shown in Fig. \ref{Fig:3DPlot} and Fig. \ref{Fig:total2DPlot}.

\begin{figure}[h]
\includegraphics[width=8cm]{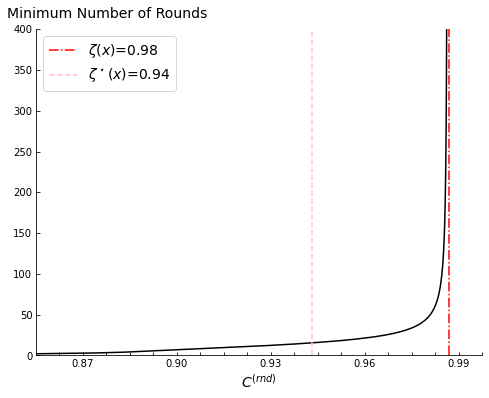}
\caption{\label{Fig:total2DPlot} Minimum number of BLOCC rounds to obtain average bipartite concurrence $C^{(\rnd)}$ from a state with $x_A:x_B:x_C=2:1:1$. Notice the gap between total random-party distillation bound $\zeta(\mbf{x})$ and $\star$-random-party distillation bound $\zeta^\star(\mbf{x})$.}
\end{figure}

\section{A Family of Maps on $\overline{\text{LOCC}}$}
The restricted Fortescue-Lo protocol described in Section \ref{Sect:Total-Random-Party-Distillation} served as an example of a map on the boundary of LOCC but not in LOCC \cite{everything}, thus proving that LOCC is not closed. 
Now we show that the operations in the step 2 of the above protocol indicate that there exist a family of maps $\{\mathfrak{J}_\gamma:0<\gamma <1\}$ on the LOCC boundary but not in LOCC, ie. $\overline{\text{LOCC}}\ni\mathfrak{J}_\gamma\notin \text{LOCC}$ for $0< \gamma <1$. \par 
We define $\mathfrak{J}_\gamma$ as follows. First fix a $0<\gamma<1$. Let Bob and Charlie both perform the measurement \begin{align}
    &M_0=\sqrt{1-\epsilon}\op{0}{0}+\op{1}{1}, &M_1=\sqrt{\epsilon}\op{0}{0}.
\end{align} If the joint outcome is one of $01$, $10$, or $11$, they stop. If the joint outcome is 00, then they repeat the same measurement. After a maximum number of $n$ iterations they stop. The $\epsilon$ and $n$ are deliberately chosen so that $(1-\epsilon)^n=\gamma$. We coarse-grain the instrument to obtain $\mathfrak{J}_{\gamma, n}=(\mathcal E_{n00}^\gamma,\mathcal E_{n01}^\gamma,\mathcal E_{n10}^\gamma,\mathcal E_{n11}^\gamma)$ where $\mathcal E_{nij}^\gamma$ include all the cases when Bob and Charlie stop upon obtaining the outcome $ij$. Thus the four maps are respectively generated by the following sets of Kraus operators: $\{M_0^n\otimes M_1^n\}, 
\{M_0^j\otimes M_1M_0^{j-1}: j=1,2,\dots,n\},
\{M_1M_0^{j-1} \otimes M_0^j: j=1,2,\dots,n\},
$ and $
\{M_1M_0^{j-1} \otimes 
M_1M_0^{j-1}: j=1,2,\dots,n\}.
$ \par 
The map $\mathfrak{J}_\gamma$ is obtained by taking $\lim_{n\rightarrow \infty} \mathfrak{J}_{\gamma,n}$. The Choi matrices $\{\Omega_{ij}^\gamma\} $ of  $\mathfrak{J}_\gamma$  can be obtained from a similar calculation as in \cite{everything}. We have\begin{align}
    &\Omega_{00}^\gamma=\begin{pmatrix}
    \gamma & \sqrt\gamma \\
    \sqrt\gamma & 1
    \end{pmatrix}^{A^\prime A}\otimes\begin{pmatrix}
    \gamma & \sqrt\gamma \\
    \sqrt\gamma & 1
    \end{pmatrix}^{B^\prime B},\\
    &\Omega_{01}^\gamma=\begin{pmatrix}
    \frac{1}{2}(1-\gamma^2) & \frac{2}{3}(1-\gamma^{3/2}) \\
    \frac{2}{3}(1-\gamma^{3/2})  & 1-\gamma
    \end{pmatrix}^{A^\prime A}\otimes\op{00}{00}^{B^\prime B},\\
    &\Omega_{01}^\gamma=\op{00}{00}^{A^\prime A}\otimes\begin{pmatrix}
    \frac{1}{2}(1-\gamma^2) & \frac{2}{3}(1-\gamma^{3/2}) \\
    \frac{2}{3}(1-\gamma^{3/2})  & 1-\gamma
    \end{pmatrix}^{B^\prime B},\\
    &\Omega_{11}^\gamma=0.
\end{align}
Thus the corresponding instruments are \begin{align}
\mathcal{E}_{ij}^\gamma:
\rho^{AB}\mapsto \tr_{A^\prime B^\prime}
\left[
\Omega_{ij}^\gamma \left(\mathbb{I}^{AB}\otimes(\rho^T)^{A^\prime B^\prime} \right)
\right]
\end{align}
It is evident that this map is on the boundary of LOCC. We argue that this map indeed does not belong to LOCC. For any $0<\gamma<1$, we consider the W-class state represented by $(x, \gamma x, \gamma x)$. The discussion in the step 2 of the previous section showed that performing the map $\mathfrak{J}_\gamma$ will optimally distill expected concurrence, thus not decreasing the value of $\zeta(\mbf{x})$. In the end the state transforms to $(\gamma^2 x, \gamma^2 x,\gamma^2 x)$. The same can never be achieved by any LOCC protocol, because we have shown that any nontrivial measurement performed by Alice or Bob will strictly decrease the expected distillable concurrence.

\subsection{EPR distillation with PPT and SEP Instruments}
So far we have discussed converting $\ket{\vec{x}}$ to an EPR state between two of the three parties using LOCC. A natural question is how this bound relates to the same probability when one only restricts the quantum instrument to being positive partial transpose (PPT) with respect to all parties. If we consider only the cases where two of the three parties end up with an EPR state or declare the outcome a failure (i.e.\ there is no outcome of distilling a W state), then we can coarse-grain the instrument to having a single round and four branches--- one for each pair of parties and one for failure. We can then express the failure probability of this instrument on an arbitrary tripartite state $\rho^{ABC}$ as a semidefinite program (SDP). 

For conciseness, let the input (resp.\ output) space be denoted $in$ (resp.\ $out$) standing for the space $ABC$ (resp.\ $A'B'C'$). For $A,B,C$, let $\Phi^{AB|C} := \op{\Phi^+}{\Phi^+}^{AB} \otimes \op{0}{0}^{C}$. For any space $B$, let $\Gamma^{B}(X^{AB}) := (\text{id}^{A} \otimes T)(X^{AB})$ where $T$ is the transpose on the space $B$. Let $\mathcal{O} := \{0,1,2,3\}$ and $\mathcal{P} = \{AA',BB',CC'\}$. Then the SDP can be expressed as
\begin{equation}\label{eqn:P_epr_sdp}
    \begin{aligned}
        \text{min.} \, & \, \tr[J_{0} (\rho^{\T} \otimes \mbb{I}^{out})] \\
        \text{s.t.} \, & \, \tr_{out}\left[\sum_{i=0}^{3} J_{i}\right] = \mbb{I}_{in} \\
        & \, \tr[\tr_{in}[J_{1}(\rho^{\T}\otimes \mbb{I}^{out})](\mbb{I}^{out} - \Phi^{AB|C})] = 0 \\
        & \, \tr[\tr_{in}[J_{2}(\rho^{\T}\otimes \mbb{I}^{out})](\mbb{I}^{out} - \Phi^{AC|B})] = 0 \\
        & \, \tr[\tr_{in}[J_{3}(\rho^{\T}\otimes \mbb{I}^{out})](\mbb{I}^{out} - \Phi^{BC|A})] = 0 \\
        & \, \Gamma^{j}(J_{i}) \geq 0 \quad \forall i \in \mathcal{O}, \forall j \in  \mathcal{P} \\
        & \, J_{i} \geq 0 \quad \forall i \in \mathcal{O} \, . 
    \end{aligned}
\end{equation}
Recalling that for a linear map $\mathcal{N}_{A \to B}$ and linear operator $X^{A}$, $\mathcal{N}(X) := \tr_{A}[J_{\mathcal{N}}(X^{\T} \otimes \mbb{I}^{B})]$ where $J_{\mathcal{N}}$ is the Choi operator, it's clear that $J_{0}$ represents the failure outcome, so the objective function is minimizing the failure probability. The next constraint guarantees coarse-graining the instrument results in a quantum channel as required. The following three constraints guarantee each success outcome is an instrument outcome that only maps the input state to an EPR state shared between two of the parties. The next line is the PPT constraints on the elements of the instruments, and the final line is simply positivity constraints on the instrument elements. Therefore we can see we have constructed an SDP for the relevant scenario.

\begin{figure}
    \centering
    \includegraphics[width=0.95\columnwidth]{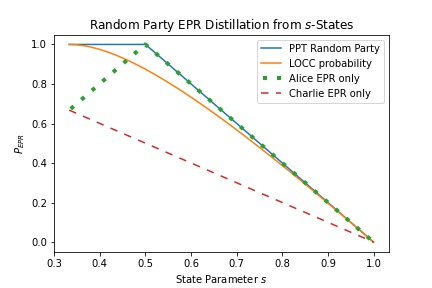}
    \caption{This depicts the gap between random party EPR distillation using a PPT (and SEP) instrument and an LOCC instrument. It also depicts the asymmetry in PPT distillation of EPR states shared between Alice and Bob or Charlie versus Charlie (resp. Bob) with Bob (resp. Charlie) or Alice. }
    \label{fig:random-party-distillation}
\end{figure}

We consider the set of $s$-states 
$$\ket{\psi_{s}}:= \sqrt{s}\ket{100} + \sqrt{\frac{1-s}{2}}(\ket{010} + \ket{001}) \ , $$
which are a specific case of $\ket{\vec{x}}$ states. Implementing the SDP \eqref{eqn:P_epr_sdp} on the $s$-states numerically using CVX \cite{cvx} with the solver MOSEK \cite{mosek} along with QETLAB \cite{qetlab} in MATLAB \cite{MATLAB:2021}, we find the gap between PPT and LOCC $P_{\text{EPR}}$ (See Fig.\ \ref{fig:random-party-distillation}), where the LOCC curve is given by $2(1-s)-(1-s)^{2}/4s$ \cite{Chitambar-2012a}. Fig.\ \ref{fig:random-party-distillation} shows a gap between LOCC and PPT success probability, as well as the optimal strategy for $s \in [1/2,1]$ establishes EPR pairs only between Alice and the other two parties. A natural further question would be if a separation between LOCC and separable operations also exists in this regime. In \cite{Chitambar-2012a}, it was shown that a SEP operations could achieve a success probability of one for $s \in [1/3,1/2]$, which means it agrees with the PPT curve in Fig. \ref{fig:random-party-distillation} in this regime. The remaining question is where a separation between LOCC and SEP exists for $s \in (1/2,1]$, as well as the possibility of a gap between SEP and PPT. We now show the above PPT curve is the same as the SEP curve. Note that Fig.\ \ref{fig:random-party-distillation}, up to numerical error, shows that for PPT operations $P_{EPR}(s) = 2(1-s)$ for $s \in [1/2,1]$. Consider $\ket{\psi_{s}}$ where $s \in [1/2,1]$. Consider Alice applying a quantum instrument to her local portion where the first outcome uses a Kraus operator $M_{0} = \op{0}{0} + \sqrt{(1-s)/s} \, \op{1}{1}$. One can verify that the state conditioned on that outcome is $\ket{\psi_{1/2}}$. Moreover, this calculation shows the probability of this outcome is $2(1-s)$. By \cite{Chitambar-2012a}, we know $\ket{\psi_{1/2}}$ can be distilled to an EPR with success probability one with SEP operations. Therefore, for $\ket{\psi_{s}}$ with $s \in [1/2,1]$, $P^{\mrm{Sep}}_{EPR} \geq 2(1-s) = P^{\mrm{PPT}}_{EPR} \geq P^{\mrm{Sep}}_{EPR}$. Thus, the PPT curve in Fig.\ \ref{fig:random-party-distillation} is the same as the SEP curve, and there exists a gap between random party EPR distillation from $\ket{\psi_{s}}$ using SEP and LOCC operations for all $s \in (1/3,1)$.

\section{Conclusion}

In this paper, we have studied how more rounds of classical communication can increase the expected bipartite entanglement yield in random-party distillation protocols.  Interestingly, the optimal number of rounds needed is highly sensitive to the type of entanglement measure considered.  For certain distillation tasks, an unbounded number of rounds is needed if the bipartite concurrent measures is optimized, whereas a finite number is needed if the $E_2$ measure is optimized.  Tight lower bounds were derived on the number of LOCC rounds required to attain a certain entanglement value.

Our work focused on W-class states since they possess a form highly amenable for analysis in the LOCC setting. In the future, it would be interesting to continue the study of round complexity for other types of multipartite entangled states.  We conjecture that high round complexity is a general feature of random-party entanglement distillation.  Intuitively, sequences of alternating weak local measurements enable a more ``gentle'' and less lossy decoupling of certain parties than when performing a hard decoupling measurement in one just one iteration of local measurements.  This leads to a greater entanglement yield when using more rounds of LOCC.  We leave it as future work to further test and develop this intuition.

\section{Acknowledgements}

E.C. thanks Hoi-Kwong Lo and Laura Man\v{c}hinska for helpful discussions in the development of this paper.  This work was supported by NSF Award No. 2016136.

%
%
%
%
\par \noindent{\LARGE\textbf{Appendices}}\
\appendix
\section{Concurrence of Bipartite W-class States}\label{appendix:a}
\par In all concurrence distillation protocols that we discuss, the bipartite outcome of any LOCC branch must be a bipartite W-class state of the form\begin{align}
    \sqrt{x_0}\ket{00}+\sqrt{x_1}\ket{01}+\sqrt{x_2}\ket{10}.
\end{align} 
This is due to the fact that the W-class is closed under LOCC. The concurrence of such state is given by $2|\sigma_1\sigma_2|$, where $\sigma_1$ and $\sigma_2$ are the Schmidt coefficients of this state. The Schmidt coefficients are the singular values of the matrix $$\begin{pmatrix}
\sqrt{x_0} & \sqrt{x_1}\\ \sqrt{x_2} &0
\end{pmatrix},$$
which are the square roots of the eigenvalues of the matrix $$\begin{pmatrix}
\sqrt{x_0} & \sqrt{x_1}\\ \sqrt{x_2} &0
\end{pmatrix}\begin{pmatrix}
\sqrt{x_0} & \sqrt{x_2}\\ \sqrt{x_1} &0
\end{pmatrix}=\begin{pmatrix}
x_0+x_1 & \sqrt{x_0x_2}\\ \sqrt{x_0x_2} & x_2
\end{pmatrix}.$$
The eigenvalues are the roots of the polynomial
$\lambda^2-(x_0+x_1+x_2)\lambda+x_1x_2=0$. From Vieta's formulas, we have $\lambda_1\lambda_2=x_1x_2$. So the concurrence is $2|\sigma_1\sigma_2|=2\sqrt{x_1x_2}.$

\section{Monotonicity of $\zeta$} 
\label{Appendix-concurrence-monotone}
\begin{lemma-Appendix}[Monotonicity part] 
The function $\zeta(\mbf{x})$ is non-increasing on average under LOCC. Furthermore, it is strictly decreasing on average when either $n_1(\mbf{x})$ or $n_2(\mbf{x})$ performs a non-trivial measurement.
\end{lemma-Appendix}

\begin{proof}
The proof is analogous to the one given in \cite{everything} which proved that $\zeta^\star$ is LOCC monotone. We can decompose any local measurement into a sequence of binary measurements \cite{AO08}\cite{OB05}. Thus each measurement can be specified by two Kraus operators $\{M_\lambda:\lambda=1,2\}$,\begin{align}
    M_\lambda:=U_\lambda\cdot\begin{pmatrix}
    \sqrt{a_\lambda} & b_\lambda \\
    0 & \sqrt{c_\lambda}
    \end{pmatrix},
\end{align}
where $a_\lambda$, $c_\lambda\geq0$, $b_\lambda\in \mathbb C$ and $U_\lambda$ is unitary (which does not affect the representation $\vec{x}$.) By the completeness relation $\sum_{\lambda=0}^1 M_\lambda^\dagger M_\lambda=I$, we have\begin{align}
    \sum_{\lambda=0}^1\begin{pmatrix}
    a_\lambda&\sqrt{a_\lambda}b_\lambda\\
    \sqrt{a_\lambda}b_\lambda^* & |b_\lambda|^2+c_\lambda
    \end{pmatrix}=\begin{pmatrix}
    1&0\\
    0&1
    \end{pmatrix},
\end{align}
thus $a_1+a_2=1$ and $c_1+c_2\leq 1$ where equality holds if and only if $b_1=b_2=0$.  \par 
To see how the state is changed by such measurement, consider when the first party Alice measures, then\begin{align}
    &\sqrt{x_0}\ket{000}+\sqrt{x_A}\ket{100}+\sqrt{x_B}\ket{010}+\sqrt{x_C}\ket{001}\nonumber \\\mapsto& 
    \frac{1}{\sqrt{p_\lambda}}[
    (\sqrt{a_\lambda x_0}+b_\lambda\sqrt{x_A})\ket{000}+\sqrt{c_\lambda x_A}\ket{100}
    \notag\\&+\sqrt{a_\lambda x_B}\ket{010}+\sqrt{a_\lambda x_C}\ket{001}
    ],
\end{align}
where $p_\lambda$ is the normalization constant that also indicates the probability of obtaining outcome $\lambda$. The post-measurement states can be represented by $\Vec{x}_\lambda=(\frac{c_\lambda}{p_\lambda}x_A, \frac{a_\lambda}{p_\lambda}x_B, \frac{a_\lambda}{p_\lambda}x_C)$ if $c_\lambda x_A\geq a_\lambda x_B$, otherwise the ordering should be changed such that $\Vec{x}_\lambda=(\frac{a_\lambda}{p_\lambda}x_B,
\frac{c_\lambda}{p_\lambda}x_A,  \frac{a_\lambda}{p_\lambda}x_C)$. More generally, when party $K\in\{A,B,C\}$ measures and outcome $\lambda$ occurs, the components of the representing vector $\Vec{x}_\lambda$ are $\{\frac{c_\lambda}{p_\lambda}x_K, \frac{a_\lambda}{p_\lambda}x_I, \frac{a_\lambda}{p_\lambda}x_J\}$ after sorting in decreasing order, where $I,J\in\{A,B,C\}\setminus \{K\}$.\par 
Define the average change in $\zeta$ incurred by the measurement as \begin{align}
    \overline{\Delta \zeta}:=p_1\zeta(\Vec{x}_1)+p_2\zeta(\Vec{x}_2)-\zeta(\Vec{x}).
\end{align}

We will show that $\overline{\Delta \zeta}\leq 0$ for all six possible cases.\\
\textbf{Case (i)} Alice measures when $x_A>x_B\geq x_C$.\par 
Since any measurement can be decomposed as a sequence of weak measurements (ie. ones where $a_1\approx c_1$ and $a_2\approx c_2$) \cite{OB05}\cite{OB06}, without loss of generality we can assume the measurement is sufficiently weak so that $c_\lambda x_A\geq a_\lambda x_B$ is guaranteed. Thus the post-measurement states can be represented by $p_\lambda \Vec{x}_\lambda=(c_\lambda x_A, a_\lambda x_B, a_\lambda x_C)$, and $\deltac$ is calculated to be\begin{align}
    \deltac=&\left(2\sqrt{x_Ax_B}+\frac{2}{3}x_C\sqrt{\frac{x_A}{x_B}}\right)
    (\sqrt{a_1c_1}+\sqrt{a_2c_2}-1)\nonumber \\
    &+\frac{1}{3}\frac{x_Bx_C}{x_A}\left(\frac{a_1^2}{c_1}+\frac{a_2^2}{c_2}-1\right).\label{eq:casei}
\end{align}
A necessary condition for $\deltac$ to be maximized is that it is maximized over allowed values of $c_2$ when $a_1$, $a_2$ and $c_1$ are fixed. Taking the partial derivative, we have 
\begin{align}
    \pd{\deltac}{c_2}=&\sqrt{\frac{a_2}{c_2}}\left(
    \sqrt{x_Ax_B}+\frac{x_C}{3}\sqrt{\frac{x_A}{x_B}}-
    \frac{x_C}{3}\frac{a_2x_B}{c_2x_A}\sqrt{\frac{a_2}{c_2}}
    \right)\nonumber \\
    \geq& \sqrt{\frac{a_2}{c_2}}\left[
    \sqrt{x_Ax_B}+\frac{x_C}{3}\left(\sqrt{\frac{x_A}{x_B}}-
    \sqrt{\frac{a_2}{c_2}}\right)
    \right]\geq 0,
\end{align}
where inequalities follow from the weakness constraint.
Thus, for fixed $a_1$, $a_2$ and $c_1$, the value of $\deltac$ is maximized when $c_2=1-c_1$. Substitute this into (\ref{eq:casei}) and evaluate around the point $(a_1, c_1)=(1/2, 1/2),$ we expand to the lowest order of $(a_1-1/2)$ and $(c_1-1/2)$ and obtain\begin{align}
    \deltac\approx -\left(\sqrt{x_Ax_B}+\frac{x_C}{3}\sqrt{\frac{x_A}{x_B}}-\frac{4}{3}\frac{x_Bx_C}{x_A}\right)(a_1-c_1)^2.
\end{align}
Let $\gamma=\sqrt{\frac{x_B}{x_A}}$, then $\gamma\leq 1$, and \begin{align}
    \deltac\approx -\left[
    \sqrt{x_Ax_B} +\frac{x_C}{3}\sqrt{\frac{x_A}{x_B}}\left(1-4\gamma^3\right)
    \right](a_1-c_1)^2.
\end{align}
If $1-4\gamma^3\geq 0$, then $\deltac\leq0$ holds.\\
If $1-4\gamma^3< 0$, then from $x_C\leq x_B$ and $\gamma \leq 1$, \begin{align}
    \deltac 
    &\leq -\sqrt{x_Ax_B}\cdot\frac{4}{3}\cdot(1-\gamma^3)(a_1-c_1)^2\leq 0,
\end{align}
with equality obtained only by the trivial measurement where $a_1=c_1$ and $a_2=c_2$. \\
\textbf{Case (ii)} Bob measures when $x_A>x_B>x_C$.\par 
WLOG, we assume that the measurement is weak enough so that $a_\lambda x_A\geq c_\lambda x_B \geq a_\lambda x_C$. Then the post-measurement states can be represented by $p_\lambda \Vec{x}_\lambda=(a_\lambda x_A, c_\lambda x_B, a_\lambda x_C)$, and $\deltac$ is calculated to be
\begin{align}
    \deltac&=\deltacs +\frac{1}{3}\frac{x_Bx_C}{x_A}(c_1+c_2-1),
\end{align}
where $\deltacs:=p_1\zeta^\star(\Vec{x}_1)+p_2\zeta^\star(\Vec{x}_2)-\zeta^\star(\Vec{x})$ with Alice being the $\bigstar$ party. Since $\zeta^\star$ have been shown to be LOCC monotone \cite{everything}, we have $\deltacs\leq0$. And since $c_1+c_2\leq 1$, $\deltac\leq0$ holds in this case with equality obtained only by the trivial measurement where $a_1=c_1$ and $a_2=c_2$.\\
\textbf{Case (iii)} Charlie measures when $x_A\geq x_B>x_C$.\par 
WLOG, we assume that the measurement is weak enough so that $a_\lambda x_B\geq c_\lambda x_C $. Then the post-measurement states can be represented by $p_\lambda \Vec{x}_\lambda=(a_\lambda x_A, a_\lambda x_B, c_\lambda x_C)$, and $\deltac$ is calculated to be
\begin{align}
    \deltac&=\deltacs +\frac{1}{3}\frac{x_Bx_C}{x_A}(c_1+c_2-1),
\end{align}
Again, $\deltac\leq0$ holds in this case. But since Charlie correspond to the party $n_2$ when calculating $\deltacs$, non-trivial measurement performed by Charlie does not necessarily decrease $\zeta^\star$ \cite{everything}, and equality holds whenever $c_1+c_2-1=0$.\\
\textbf{Case (iv)} Bob measures when $x_A> x_B=x_C$.\par 
WLOG, we assume that $a_1\leq c_1$ (thus $a_2\geq c_2$.) We also assume that the measurement is weak enough so that $a_1 x_A\geq c_1 x_B $. Then the post-measurement states can be represented by \begin{align}
    &p_1 \Vec{x}_1=(a_1 x_A, c_1 x_B, a_1 x_B),\\
    &p_2 \Vec{x}_2=(a_2 x_A, a_2 x_B, c_2 x_B),
\end{align} and $\deltac$ is calculated to be
\begin{align}
    \deltac&=\deltacs +\frac{1}{3}\frac{x_Bx_C}{x_A}(c_1+c_2-1),
\end{align}
Again, $\deltac\leq0$ holds in this case with equality obtained only by the trivial measurement where $a_1=c_1$ and $a_2=c_2$.\\
\textbf{Case (v)} Alice measures when $x_A= x_B>x_C$.\par 
WLOG, we assume that $a_1\leq c_1$ (thus $a_2\geq c_2$.) We also assume that the measurement is weak enough so that $c_2 x_A\geq a_2 x_B $. Then the post-measurement states can be represented by \begin{align}
    &p_1 \Vec{x}_1=(c_1 x_A, a_1 x_A, a_1 x_C),\\
    &p_2 \Vec{x}_2=(a_2 x_A, c_2 x_A, a_2 x_C),
\end{align} and $\deltac$ is calculated to be
\begin{align}
    \deltac&=\deltacs +\frac{1}{3}\frac{x_C}{c_1}(a_1^2-c_1^2),
\end{align}
Again, $\deltac\leq0$ holds in this case with equality obtained only by the trivial measurement where $a_1=c_1$ and $a_2=c_2$.\\
\textbf{Case (vi)} Alice measures when $x_A= x_B=x_C=x$.\par 
WLOG, we assume that $a_1\leq c_1$ (thus $a_2\geq c_2$.) Then the post-measurement states can be represented by \begin{align}
    &p_1 \Vec{x}_1=(c_1 x, a_1 x, a_1 x),\\
    &p_2 \Vec{x}_2=(a_2 x, a_2 x, c_2 x),
\end{align} and $\deltac$ is calculated to be
\begin{align}
    \deltac&=\deltacs +\frac{1}{3}\frac{x_C}{c_1}(a_1^2-c_1^2),
\end{align}
Again, $\deltac\leq0$ holds in this case with equality obtained only by the trivial measurement where $a_1=c_1$ and $a_2=c_2$.\\
\end{proof}
\bibliography{bib}
\bibliographystyle{plainurl}
\end{document}